%% file: bsc_25.tex
\documentclass[onecolumn]{IEEEtran}
\usepackage{rotating}
\usepackage{tabularx}
\usepackage{multirow}
\usepackage{tikz}
\usetikzlibrary{calc, snakes}
\tikzset{line width=200pt}
\usepackage[utf8]{inputenc}

\usepackage{amsthm}
\usepackage{amsmath,amsfonts,amssymb,latexsym,cite}

\usepackage{enumerate}
\usepackage{hyperref} 
\usepackage{graphicx} 
\usepackage{color}
\usepackage{epsfig}
\usepackage{graphicx}

\usepackage{latexsym}
\usepackage{amsmath}
\usepackage{amssymb}

\newcommand{\suppress}[1]{}

\newtheorem{theorem}{Theorem}[section]
\newtheorem{lemma}[theorem]{Lemma}

\newtheorem{corollary}[theorem]{Corollary}

\newtheorem{remark}[theorem]{Remark}

\newcommand{\type}[1]{{\theta_{#1}}}

\newcommand{\ie}{{i.e.}}
\newcommand{\eg}{{e.g.}}

\newcommand{\pee}{p}

\newcommand{\Enc}{Enc}

\newcommand{\mess}{u}
\newcommand{\bmess}{{U}}
\newcommand{\ind}{t}

\newcommand{\cV}{{\cal{V}}}
\newcommand{\cW}{{\cal{W}}}
\newcommand{\ogs}{{\mathcal M}_{og}}

\newcommand{\code}{{\cC}}

\def\bsc{$\mbox{BSC}_p$}
\def\bscq{$\mbox{BSC}(q)$}

\mathchardef\mhyphen="2D
\def\wcefII{WCEF\mhyphen II}

\newcommand{\bx}{{{x}^\bl}}
\newcommand{\by}{{{y}^\bl}}
\newcommand{\bz}{{{z}^\bl}}
\newcommand{\bs}{{{s}^\bl}}

\def\evcf{{\cal{E}}_{C1}}
\def\evcs{{\cal{E}}_{C2}}
\def\evo{{\cal{E}}_{O}}
\def\evd{{\cal{E}}_{D}}

\def\eone{{\epsilon_2}}
 \def\etwo{{\epsilon_1}}

\def\cX{\mbox{$\cal{X}$}}
\def\cY{\mbox{$\cal{Y}$}}

\def\cZ{\mbox{$\cal{Z}$}}
\def\cB{\mbox{$\cal{B}$}}
\def\cC{\mbox{$\cal{C}$}}

\def\cS{\mbox{$\cal{S}$}}

\def\cZ{\mbox{$\cal{Z}$}}

\def\cT{\mbox{$\cal{T}$}}

\def\e{\varepsilon}

\def\type{\mbox{type}}

\def\bX{{X^\bl}}
\def\bY{{Y^\bl}}
\def\bZ{{Z^\bl}}
\def\bS{{S^\bl}}

\def\bl{{n}}
\def\rate{{R}}
\def\Ber{\mbox{Ber}}

\def\01{\{0,1\}}

\newcolumntype{L}[1]{>{\raggedright\let\newline\\\arraybackslash\hspace{0pt}}m{#1}}
\newcolumntype{C}[1]{>{\centering\let\newline\\\arraybackslash\hspace{0pt}}m{#1}}
\newcolumntype{R}[1]{>{\raggedleft\let\newline\\\arraybackslash\hspace{0pt}}m{#1}}

\begin{document}

\IEEEoverridecommandlockouts
\title{Sufficiently Myopic Adversaries are Blind
\thanks{A preliminary version of this work appeared in \cite{dey2015sufficiently}.}
}
\date{}

\author{Bikash Kumar Dey, Sidharth Jaggi, Michael Langberg}

\date{}
\maketitle
\begin{abstract}
In this work we consider a communication problem in which a sender, Alice, wishes to communicate with a receiver, Bob, over a channel controlled by an adversarial jammer, James, who is {\em myopic}. Roughly speaking, for blocklength $n$, the codeword $\bX$ transmitted by Alice is corrupted by James who must base his adversarial decisions (of which locations of $\bX$ to corrupt and how to corrupt them) not on the codeword $\bX$ but on $\bZ$, an image of $\bX$ through a noisy memoryless channel. More specifically, our communication model may be described by two channels. A memoryless channel $p(z|x)$ from Alice to James, and an {\it Arbitrarily Varying Channel} from Alice to Bob, $p(y|x,s)$ governed by a state $\bS$ determined by James.  In standard adversarial channels the states $\bS$ may depend on the codeword $\bX$, but in our setting $\bS$ depends only on James's view $\bZ$. 

The myopic channel captures a broad range of channels and bridges between the standard models of memoryless and adversarial (zero-error) channels. In this work we present upper and lower bounds on the capacity of myopic channels. For a number of special cases of interest we show that our bounds are tight. We extend our results to the setting of {\em secure} communication in which we require that the transmitted message remain secret from James. For example, we show that if (i) James may flip at most a $p$ fraction of the bits communicated between Alice and Bob, and (ii) James views $\bX$ through a binary symmetric channel with parameter $q$, then once James is ``sufficiently myopic" (in this case, when $q>p$), then the optimal communication rate is that of an adversary who is ``blind'' (that is, an adversary that does not see $\bX$ at all), which is $1-H(p)$ for standard communication, and $H(q)-H(p)$ for secure communication.
A similar phenomenon exists for our general model of communication.

\end{abstract}


\section{Introduction}

In the study of point-to-point communication, a sender Alice wishes to transmit a message $\bmess$ to a receiver Bob over a noisy channel governed by a jammer James. To do so, she encodes $\bmess$ into a length-$\bl$ vector $\bX$ and transmits it over the channel, resulting in the received word $\bY$. 
Two types of channel models that have seen significant attention over the last decades are the memoryless channel model, e.g., \cite{shannon_mathematical_1949} in which the channel is governed by a conditional distribution $p(y|x)$ which is completely oblivious \cite{Langberg:08} (or ``blind'') of the message $\bX$ being transmitted and the adversarial (omniscient) channel model in which James is thought of as an adversarial entity who can maliciously design the error imposed to fit the specific codeword transmitted, \cite{hamming1950error}.
While the capacity of the former model is well-understood, that of the latter encompasses numerous  open problems in coding and information theory. 
This state of affairs has lead to the study of several channel models that conceptually lie between the two extreme communication models, those in which the channel is oblivious of the 
transmitted codeword $\bX$ and those in which the channel acts as an adversarial jammer.
These include arbitrarily varying channels (AVCs), \eg~\cite{blackwell_capacities_1960,ahlswede1978elimination,csiszar1988capacity,csiszar1991capacity}, causal channels, \eg~\cite{dey_codes_2009,langberg_binary_2009,dey_coding_2010,haviv_beating_2011,dey_improved_2012,dey_codes_2013,dey_upper_2013,chen_characterization_2014},  and computationally limited channels, \eg~\cite{guruswami_codes_2010}.

Inspired by the study of Sarwate~\cite{sarwate_coding_2010}, in this work we consider the model of {\em myopic} adversarial jammers.
In the myopic setting, the jammer James is still a malicious entity that wishes to carefully design his error to corrupt communication, however his view of the codeword $\bX$ is limited in the sense that it is masked through a noisy memoryless channel $p(z|x)$. 
If the channel between Alice and James is of full rate, the myopic model reduces to that of the standard omniscient adversarial model, and if it is of zero rate, the myopic model captures the model of a ``blind'' (or ``oblivious") adversary that has no knowledge whatsoever on the codeword $\bX$ transmitted.

Formally, the myopic model is described by two channels. A memoryless channel $p(z|x)$ from Alice to James, and an AVC from Alice to Bob.
The AVC is modeled by a state channel $p(y|x,s)$, where the vector of states $\bS$ (one state for each time step) is determined by James as a function of his masked view $\bZ$ of the transmitted codeword $\bX$.
See Figure~\ref{fig:model}.

\tikzset{line/.style={draw, -latex',shorten <=1bp,shorten >=1bp}}

\begin{figure}[t]
\setlength{\unitlength}{1cm}
\centering
\begin{tikzpicture}[trim left, scale=1]

\node (alice) [draw, rectangle] at (0.5,3.25) {Alice};
\node (bob) [draw, rectangle] at (7.5,3.25) {Bob};
\node (bobch) [draw, rectangle] at (5, 3.25) {$p(y|x,s)$};
\node (calch) [draw, rectangle] at (2.5,5) {$p(z|x)$};
\node (James) [draw, rectangle] at (5,5) {James};

\draw [->] (alice) -- (1.5,3.25) node[below] (tx) {$x^n$} -- (bobch);
\draw [->] (tx) -- (1.5,5) -- (calch);
\draw [->] (calch) -- +(1.25,0) node[below] {$z^n$} -- (James);
\draw [->] (James) --
          +(0,-0.75) node[right] {$s^n$} -- (bobch);
\draw [->] (bobch) -- +(1.5,0) node[below] {$y^n$} -- (bob);

\end{tikzpicture}
\caption{The myopic channel model}
\label{fig:model}
\end{figure}
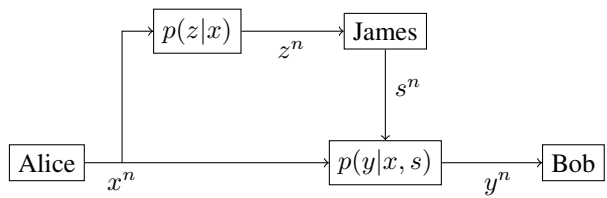

In this work we study the capacity of myopic adversarial channels.
We start by studying a natural binary myopic channel in which (i) James may flip at most a $p$ fraction of the bits communicated between Alice and Bob, and (ii) James views $\bX$ through a binary symmetric channel with parameter $q$ (i.e., $BSC(q)$). Namely, in our notation, the Hamming weight of $\bS$ is at most $pn$, $p(z|x)=q$ for $z \neq x$, and $p(y|x,s)=1$ iff $y=x+s$ (and otherwise 0). We aim to characterize the capacity of the channel under varying values of $q$, our limitation on the noise level to James.
When $q=0$, namely when James has full knowledge of the codeword $\bX$, the channel reduces to the omniscient adversarial channel for which the capacity is a central open problem in coding theory and only upper and lower bounds on capacity exist \cite{gilbert_comparison_1952, varshamov_estimate_1957,mceliece_new_1977}. When $q=1/2$, namely when James is blind, it is shown in~\cite{Langberg:08, csiszar1988capacity} that the capacity equals that of the channel in which James flips bits randomly, i.e. the $BSC(p)$, which equals $1-H(p)$.

The focus of this work is in the study of intermediate values of $q$.
In a nutshell, we present a dichotomous behavior of the channel. 
If James is ``sufficiently myopic" then the optimal communication rate is that of a blind James, namely $1-H(p)$. Specifically, we show that an optimal rate of $1-H(p)$ is achievable as long as $q>p$.
If on the other hand $q < p$, then the capacity of the myopic channel equals that of the omniscient channel,\footnote{To be precise, the above dichotomous behavior is proven to hold for deterministic codes.
For codes that allow randomness at encoder (which is {\em not} shared with the receiver), known as stochastic codes, we leave open the question whether one can obtain rates higher than those of the omniscient adversary for the case $q<p$.} which is known to be bounded away from $1-H(p)$ for all $p$, and in fact equals $0$ for all $p>1/4$.
We extend our results to the setting of secure communication in which one requires that the transmitted message remain secret from James. In this extended setting we show a similar phenomena: as long as $q>p$ the capacity  equals that obtained for blind adversaries (which is $H(q)-H(p)$).

We then turn to study the myopic model in its full generality, for general memoryless channels $p(z|x)$ connecting Alice and James, and general state channels $p(y|x,s)$ connecting Alice and Bob.
For the general setting we obtain upper and lower bounds on capacity, both in the standard setting of communication, and in that of secure communication.
As an additional case study, we study the setting in which the channel to James is a binary erasure
channel $BEC(q)$, and James can erase up to a fraction $p$ of the transmitted
bits observed by Bob. For this special case, through a refinement of our
arguments, we show that the capacity is $1-p$ if $q>p$; and for
$q<p$, the deterministic capacity is the same as that for an omniscient 
adversary.
We also consider some more general binary input adversarial channels to
study communication rates as well as secrecy rates. In these channels,
James can erase as well as flip some fractions of bits. His own observation
may be over an arbitrary binary input channel.

As mentioned above, the work most relevant to ours is that of Sarwate~\cite{sarwate_coding_2010} in which the myopic channel model is studied under the assumption that Alice and Bob hold shared randomness that is not known to James (i.e., under the assumption of randomized coding). In this setting, a single-letter characterization to the randomized coding capacity is obtained. 
As with our study, the results in \cite{sarwate_coding_2010} bridge between the randomized capacity when the adversary James is assumed to be blind and that when James has full knowledge of the codeword transmitted.

Although our study was inspired by, and builds on, that of \cite{sarwate_coding_2010}, it differs from \cite{sarwate_coding_2010} in two important aspects. Primarily, and most importantly, we study the case of deterministic codes (in which there is no shared randomness between Alice and Bob). The study of deterministic codes introduces many challenges that do not exist in the case of randomized codes, and involves a new set of analytical tools in its analysis.
Secondly, we study the general case in which the codewords $\bX$ of Alice and the state space $\bS$ of James are constrained. 
Our enhanced setting was explicitly left open in \cite{sarwate_coding_2010}.

Another model related to our work is the study of the wiretap channel of type II
with an active eavesdropper. 
Aggarwal et al.~\cite{aggarwal2009wiretap} considers a
model with an adversary who can choose
a $p$ fraction of bits to observe and also erase these bits. They showed
that any rate upto $1-p-H(2p)$ can be achieved. In our notation, their model 
has $q=1-p$
fraction of erasures in James' channel. If James experiences random erasures,
then Theorem~\ref{th:erasuresecrecy} guarantees a secrecy rate
upto $(1-p) - p = 1-2p$. However, in \cite{aggarwal2009wiretap}  James has the
additional power of choosing which bits to observe. 
As a special case of Theorem~\ref{th:wiretaptwo}, we are able to obtain rate $1-2p$ on the model of \cite{aggarwal2009wiretap} as well (see Remark~\ref{rem:erasure-erasure}).
Additional works that address the action of myopic adversaries include 
\cite{wang_on_capacity_2016} which considers the study of the wiretap channel of type II with an active eavesdropper that can flip bits. Theorem~\ref{th:wiretaptwo} generalizes their main result to an active eavesdropper who can erase as well as flip bits. The work in \cite{molavianjazi2009arbitrary}, studies a different model of active myopic adversaries in which  there are two non-cooperating adversarial entities, the Eavesdropper and the Jammer. In a nice sequence of works by Boche\footnote{From whom we also borrow the idea of calling the jammer James.} et al\cite{bjelakovic2013secrecy,notzel2016arbitrarily,wiese2014arbitrarily} the problem of secure communication in the presence of a myopic jammer is also considered, but in general in the models considered either common randomness between the transmitter and the receiver is necessary, or only multi-letter capacity characterizations are derived (or both). Also, \cite{shafiee_mutual_2009, sarwate_avc_2012} consider myopicity in the context of AWGN channels.
For specific channels over sufficiently large alphabets on which the attacks are either additive (e.g., \cite{wang2016model,safavi2015model,wang2014efficient,safavi2013codes,safavi2013efficient,wang2015erasure}, summarized in~\cite{safavi2015secure}), or ``overwrite'' (e.g., \cite{zhang2015coding, zhang2015talking}, summarized in~\cite{kadhe2015reliable}), more is known; computationally efficient codes meeting information-theoretic bounds are known. See Table~\ref{table} for a summary of previous related work. 

Our paper is structured as follows.
In Section~\ref{sec:model} we give a precise model for the myopic setting.
In Section~\ref{sec:results} we state our main results (which are also summarized in Table~\ref{table}).
Our results are first presented for the special binary symmetric error case discussed above and then in full generality.
We also discuss a refinement of the general arguments to
an erasure-erasure channel, and other binary input channels with
erasing and flipping adversary.
Section~\ref{sec:bin} presents the proof of the main result for the binary 
case. Section~\ref{app:gen} presents the proof of the lower bound for 
the general model.

\input{table.tex}

\section{Model}
\label{sec:model}

The myopic channel is defined by its input alphabet $\cX$, output alphabet to James $\cZ$, state alphabet $\cS$, output alphabet to Bob $\cY$,  probability distribution 
 for the channel connecting Alice and James $p(z|x)$,  probability distribution $p(y|x,s)$ for the channel connecting Alice and Bob, the state constraint $\cW$, and the input constraint $\cV$.
The three parties of the channel, Alice, Bob, and James are described below (see Figure~\ref{fig:model}).

\noindent {\bf{Alice's encoder:}}
Alice has a {\it message} $\bmess$ uniformly distributed in 
$\{0,1\}^{\bl \rate}$ that she wants to
transmit to Bob; $\rate$ denotes the {\it rate} of her message, and $\bl$
the {\it block-length} of Alice's transmissions. To effect this
communication, Alice encodes her message using an encoder $\Enc : \{0,1\}^{\bl \rate} \rightarrow \cX^n$ to
output a {\it transmitted vector} $\bX = \Enc(\mess)$.
We emphasize that Alice's encoder is deterministic.
The encoder has to satisfy the constraint $\type(\bX) \in \cV$, where $\cV$ is a set of types over the alphabet $\cX$.

\noindent {\bf{Channel from Alice to James:}}
James observes the output of $\bX$ passing through a memoryless channel $p(z|x)$.
More precisely, the channel law is given as
$Pr(\bZ|\bX)=\prod_\ind p(z_\ind|x_\ind)$.
Based on James's non-causal observation $\bZ$, he chooses a length-$\bl$ 
{\it state vector} $\bS$.
The state vector $\bS$ is restricted to have $\type(\bS) \in \cW$, where $\cW$ is a set of types over the alphabet $\cS$.

\noindent {\bf{Channel from Alice to Bob:}}
Bob observes the output
$\bY$ obtained through the channel  $p(y|x,s)$.
More precisely, for state $\bS = (s_1,\dots,s_\bl)$ the channel law is given as
$Pr(\bY|\bX,\bS)=\prod_\ind p(y_\ind|x_\ind,s_\ind)$.

\noindent {\bf{Successful communication:}}
Given $\bY$, Bob decodes a message $\hat{u}\in\{0,1\}^{Rn}$.
Communication is considered successful if the transmitted message $u$ equals $\hat{u}$.
The average error in communication is defined as
$\e = \frac{1}{2^{Rn}}\sum_u Pr(u\ne \hat{u})$.\footnote{Notice that in the setting of deterministic code design the average error criteria is essential for the study of the myopic model (in which we assume that James bases his decisions on a corrupted view of $\bX$),
as otherwise, in the study of maximum error, James may neglect $\bZ$ and focus his strategy on a single transmitted codeword, yielding the channel $p(z|x)$ irrelevant to the study of capacity.
This state of affairs does not hold once stochastic coding is considered.
Connections exist between the study of deterministic codes under the average error criteria and stochastic codes under the maximum error criteria in the context of AVCs, e.g., \cite{ahlswede1978elimination}.
In this work we focus on deterministic codes (which we prove are optimal for several of the settings we study).}
Rate $R$ is achievable over the myopic channel if for any $\e>0$ there exists a block length $\bl$ such that the average error in communication is at most $\e$.
The channel capacity is the supremum of all achievable rates.

\noindent {\bf Secrecy:} At times we will study the secrecy (i.e., secure)
capacity between Alice and Bob. In this setting, in addition to correct
decoding, we require that James's view $\bZ$ be almost
independent of Alice's message $u$, namely that $\frac{1}{n}I(\bZ;U)<\e$.

\section{Our results}
\label{sec:results}

In what follows we present our results.
The results are presented first for the special binary symmetric error 
myopic channel (Sec.~\ref{sec:binaryerror}) discussed in the Introduction, and then in generality~(Sec.~\ref{sec:general}).
We also present refinements of our general results for
a binary erasure-erasure model, and more general binary AVC where James can
erase and flip some fractions of transmitted bits in Subsection~\ref{sec:erasure} and Subsection~\ref{sec:eraseflip} respectively.
\footnote{All our results on bit-flip and/or erasure channels can be generalized to
$q$-ary additive and erasure channels, i.e., where James can erase some
fraction of symbols and/or add arbitrary symbols of his choice to some
fraction of symbols. In the interest of brevity we do not present these straightforward generalizations.}

\subsection{The myopic binary $C(q,p)$ channel}
\label{sec:binaryerror}
Our studies begin with the binary channel $C(q,p)$ (Fig.~\ref{fig:ot_hbc_wtap_bcast}) characterized by the pair of parameters $(q,p)$ in which (i) James views $\bX$ through a binary symmetric channel with parameter $q$ (i.e., $BSC(q)$), and (ii) James may flip at most a fraction $p$  of the bits communicated between Alice and Bob. Namely, in our notation, we set $\cX =\cZ=\cY=\cS = \{0,1\}$, $p(z|x)=q$ for $z \neq x$, $p(y|x,s)=1$ iff $y=x+s$ (and otherwise 0), and $\cW = \{(1-p',p')|p'\leq p\}$ (i.e., $\type(\bS)\in \cW$ if and only if $\|\bS\| \leq pn\}$ where $\| \cdot \|$ denotes the Hamming weight).

\tikzset{XOR/.style={draw,circle,append after command={
        [shorten >=\pgflinewidth, shorten <=\pgflinewidth,]
        (\tikzlastnode.north) edge (\tikzlastnode.south)
        (\tikzlastnode.east) edge (\tikzlastnode.west)
        }
    }
}
\tikzset{line/.style={draw, -latex',shorten <=1bp,shorten >=1bp}}

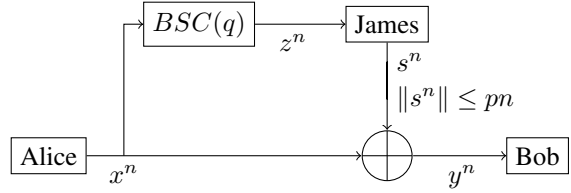
\begin{figure}
\setlength{\unitlength}{1cm}
\centering
\begin{tikzpicture}[trim left, scale=1]

\node (alice) [draw, rectangle] at (0.5,3.25) {Alice};
\node (bob) [draw, rectangle] at (7,3.25) {Bob};
\node (bobch) [XOR, scale=2] at (5, 3.25) {};
\node (calch) [draw, rectangle] at (2.5,5) {$BSC(q)$};
\node (James) [draw, rectangle] at (5,5) {James};

\draw [->] (alice) -- (1.5,3.25) node[below] (tx) {$x^n$} -- (bobch);
\draw [->] (tx) -- (1.5,5) -- (calch);
\draw [->] (calch) -- +(1.25,0) node[below] {$z^n$} -- (James);
\draw [->] (James) -- 
          +(0,-1) node[right] {$\|s^n\| \leq pn$} -- 
          +(0,-0.5) node[right] {$s^n$} -- (bobch);
\draw [->] (bobch) -- +(1,0) node[below] {$y^n$} -- (bob);

\end{tikzpicture}
\caption{The binary symmetric error channel $C(q,p)$}
\label{fig:ot_hbc_wtap_bcast}
\end{figure}

We first study the case $q>p$:

\begin{theorem}
\label{th:binarycapacity}

For $q>p$, the capacity of the binary myopic adversarial channel
$C(q,p)$ is $1-H(p)$.
The capacity is achieved by random
codes with input distribution $\Ber(1/2)$.
\end{theorem}

To prove Theorem~\ref{th:binarycapacity} we must present both an upper and a lower bound on capacity.
The upper bound is relatively simple and follows from the fact that James
may roughly mimic a memoryless $BSC(p)$ (no matter what $q$ is).
Specifically, James can completely neglect his view $\bZ$ and just construct a state vector uniformly at random among those with type $(1-p,p)$.
The converse of the channel coding theorem now shows that the rate in this case is bounded above by $1-H(p)$.
Our main contribution in the study of $C(q,p)$ is in the achievability part of Theorem~\ref{th:binarycapacity} in which we show that one can obtain rates arbitrarily close to $1-H(p)$.
The technical proof as well as a proof outline are given in Section~\ref{sec:bin}.

We next study the case of $q<p$.
Here, we show that the capacity equals that of the omniscient adversary.
\begin{theorem}
\label{th:binaryomni}
For $q<p$, the deterministic coding capacity of the binary myopic 
adversarial channel $C(q,p)$ is the same as that of the binary 
adversarial channel with an omniscient adversary.
\end{theorem}

\begin{proof}
We assume successful communication at rate $R$ over $C(q,p)$ and show that $R$ is achievable in the omniscient channel model as well.
Consider the code that allows communication at rate $R$.
The same code must also allow communication at rate $R$ over $BSC(q)$ (this
follows from the fact that $q<p$ and thus James can roughly mimic $BSC(q)$, just as 
described above in the converse to Theorem~\ref{th:binarycapacity}). 
Since for such an adversarial action, Bob can still decode $\bX$,
this implies that James,
who views $\bX$ through a $BSC(q)$ is able to decode $\bX$ as well,
implying, in turn, that James is actually omniscient.
\end{proof}

We finally turn to study the context of secure communication.
Here, we first consider the binary symmetric broadcast channel with
independent BSC to Bob and James with cross-over probabilities $p$
and $q$ respectively. 
Then it is well known~\cite{csiszar_broadcast_1978} that the message transmission
capacity to Bob under the secrecy condition 
is $H(q)-H(p)$.  
An achievability scheme in this case is to append the $nR$ bits of
message with $n(1-H(q))$ bits of private randomness and encode the resulting
string with a random channel code of rate $1-H(p)$.

In our channel $C(q,p)$, for $q>p$, the secrecy capacity is also $H(q)-H(p)$.
The encoding can be done in the same way as before: appending random 
bits to the message and then encoding using a random code. James
can not learn anything about the message by the secrecy results
in the random channel case discussed above. Since James is 
sufficiently myopic, by Theorem~\ref{th:binarycapacity}, Bob
can decode the message and the private randomness irrespective
of James's strategy. So, we have

\begin{theorem}
For $q>p$, the binary myopic channel $C(q,p)$ has secrecy capacity
$H(q)-H(p)$.
\label{th:binarysecrecy}
\end{theorem}

\subsection{General myopic channels}
\label{sec:general}
We now turn to present our results for the myopic model in full generality.
We consider the setup where James's channel is given by $p_{Z|X}$,
and his state $\bS$ is constrained to have a type in the set $\cW$. We
denote this channel by $C(p_{Z|X}, \cW)$.
To obtain single letter upper and lower bounds on the capacity of myopic channels we consider the types of the vectors $\bX$, $\bZ$, $\bS$, and $\bY$, and certain distributions on them.
Our achievability scheme uses a random code governed by the single letter distribution $p_X \in \cV$.
Let $p_{Z|X}$ be James's 
channel law (we now explicitly specify the channel as a subscript to avoid confusion).
The distributions, $p_X$ and $p_{Z|X}$ give rise to a joint distribution $p_{XZ}$ and a marginal 
distribution $p_{Z}$. Recall that  $\cW$ is the set of state types of $\bS$ which James may impose.
Let $\cW_{S|Z}$ be the set of conditional distributions $p_{S|Z}$ which
results in a marginal distribution $p_S$ in the set $\cW$.
Namely, $p_{S|Z}$ is in $\cW_{S|Z}$ if and only if 
$$
p_S(\cdot) = \sum_{x,z} p_X(x)p_{Z|X}(z|x)p_{S|Z}(\cdot|z) \in \cW
$$
Finally, we use $p_{Y|XS}(y|x,s)$ which is given as part of the channel definition.
Note that $p_X$ and $p_{S|Z}$ define a joint single letter distribution over random variables $X,Z,S,Y$ defined by  
$p_Xp_{Z|X}p_{S|Z}p_{Y|XS}$. 

We are now ready to present our results for the general model.
Our first theorem addresses achievability. 
For technical reasons, we focus only on ``state-deterministic channels'', \ie~where $\bY$ is a deterministic function of $(\bX,\bS)$.
We elaborate on channels which are not state-deterministic in Section~\ref{app:gen}.

\begin{theorem}
\label{th:g_lower}
For state-deterministic $\pee_{Y|XS}$ a rate $R$ is achievable if there exists a $p_X \in \cV$ such that
\begin{align*}
& (a)\ \forall p_{S|Z} \in \cW_{S|Z}:\  R < I(X;Y).\\
& (b)\ \forall p_{S|Z} \in \cW_{S|Z}:\ I(X;Z) < I(X;Y).\\
& (c)\max_{p_{S|Z} \in \cW_{S|Z}}H(X|Y,S) + \max_{p_{S|Z} \in 
\cW_{S|Z}}H(Y|X) < H(X|Z).
\end{align*}
Such a rate is achievable using random codes generated using
the input distribution $p_X$.
\end{theorem}

Our proof of Theorem~\ref{th:g_lower} is along similar lines as that of
Theorem~\ref{th:binarycapacity}, and is presented in Sec.~\ref{app:gen}. 
The theorem guarantees the maximum rate $I(X;Y)$ (condition (a)) 
against an oblivious adversary with state constraint, provided
the state constraint satisfies the two myopicity conditions (b) and
(c). Here condition (b) says that James's channel should be worse
than Bob's channel, i.e., James's view should be less `informative' than
Bob's. This corresponds to the condition $q>p$ in the binary case in 
Theorem~\ref{th:binarycapacity}. Though condition (c) also says 
something similar
in nature, its exact form is not intuitive. It comes due to
a technical requirement in the proof (see Remark~\ref{rem:erasure-erasure}).

We now give an upper bound obtained
by considering only memoryless feasible jamming strategies. The proof
is obvious, and is omitted.

\begin{theorem}
\label{th:g_upper}
A rate $R$ is achievable only if there exists a $p_X \in \cV$ such that
$\ \forall p_{S|Z} \in \cW_{S|Z}:\  R < I(X;Y). $
\end{theorem}

Using a similar argument as in the binary case
(Theorem~\ref{th:binarysecrecy}),
we obtain the following achievability result for secrecy rates.

\begin{theorem}
\label{th:g_secrecy}
For state-deterministic $\pee_{Y|XS}$ a secrecy rate $R$ is achievable if there exists a $p_X \in \cV$ such that
\begin{align*}
& (a)\ \forall p_{S|Z} \in \cW_{S|Z}:\  R < I(X;Y) - I(X;Z).\\
& (b)\ \max_{p_{S|Z} \in \cW_{S|Z}}H(X|Y,S) + \max_{p_{S|Z} \in 
\cW_{S|Z}}H(Y|X) < H(X|Z).
\end{align*}
\end{theorem}

\begin{remark}
The proof of secrecy for rates under $I(X;Y) - I(X;Z)$ is following
similar encoding and arguments as that for wiretap channels. For wiretap
channels, it is known that this can be improved to the maximum
of $I(U;Y) - I(U;Z)$ over all $p(u)p(x|u)$. However, our proof
(of Theorem~\ref{th:g_secrecy}) uses a state-deterministic $\pee_{Y|XS}$.
Introducing an auxiliary random variable $U$ will result in
an effective probabilistic AVC $\pee_{Y|US}$, for which our present
proof does not hold. Please also see footnote~\ref{foo:prob}.
\end{remark}

\subsection{Binary erasure-erasure channels}
\label{sec:erasure}
Communicating securely in the presence of an active eavesdropper has attracted some attention in the recent literature (e.g.~\cite{aggarwal2009wiretap,wang_on_capacity_2016}). Hence in this and the next subsection we remark on the implications of our techniques/results for some binary-input channels. One of the challenges is that unlike in the binary symmetric case (Theorem~\ref{th:binarycapacity}), for general myopic channels, our
lower bound (Theorem~\ref{th:g_lower}) does not meet the upper bound
(Theorem~\ref{th:g_upper}). This is due to the difficulty of finding
a single-letter expression for a counting argument in the proof of 
Lemma~\ref{lem:genlist1}.

For an erasure-erasure channel
(referred to as $CE(q,p)$) where James's channel is a $BEC(q)$,
and he can erase at most a $p$ fraction of the transmitted bits,
Theorem~\ref{th:g_lower} guarantees rates upto $1-p$ only if
$q > p+H(p)$, whereas the upper bound of $1-p$ in Theorem~\ref{th:g_upper}
is valid whenever $q>p$. This gap can be eliminated by careful analysis using specific properties of erasure channels.
As a result, we have the capacity results as given below. 

\begin{figure}[h]
\setlength{\unitlength}{1cm}
\centering
\begin{tikzpicture}[trim left, scale=1]

\node (alice) [draw, rectangle] at (0.5,3.25) {Alice};
\node (bob) [draw, rectangle] at (8.5,3.25) {Bob};
\node (bobch) [draw, rectangle] at (5, 3.25) {$Y=\begin{cases}
X; & \text{if } S=0\\ \perp; & \text{if } S=1\end{cases}$};
\node (calch) [draw, rectangle] at (2.5,5) {$BEC(q)$};
\node (James) [draw, rectangle] at (5,5) {James};

\draw [->] (alice) -- (1.5,3.25) node[below] (tx) {$x^n$} -- (bobch);
\draw [->] (tx) -- (1.5,5) -- (calch);
\draw [->] (calch) -- +(1.25,0) node[below] {$z^n$} -- (James);
\draw [->] (James) --
          +(0,-0.6) node[right] {$s^n, \|s^n\| \leq pn$} -- (bobch);
\draw [->] (bobch) -- +(2.25,0) node[below] {$y^n$} -- (bob);

\end{tikzpicture}
\caption{The binary erasure-erasure adversarial channel $CE(q,p)$}
\label{fig:erasure-erasure}
\end{figure}
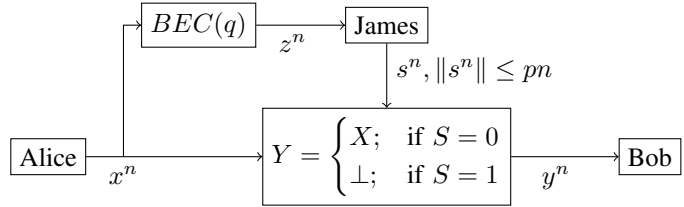

\begin{theorem}
\label{th:erasurecapacity}
For $q>p$, the capacity of the binary erasure-erasure channel
$CE(q,p)$ is $1-p$.
The capacity is achieved by random
codes with input distribution $\Ber(1/2)$.
\end{theorem}

The proof of this result follows as a special case of 
(Theorem~\ref{th:g_lower}) with a specific refinement in Lemma~\ref{lem:genlist1} as discussed
in Remark~\ref{rem:erasure-erasure}.

The next two theorems follow using similar arguments as those of
Theorem~\ref{th:binaryomni} and Theorem~\ref{th:binarysecrecy} respectively.
The proofs are thus omitted.
\begin{theorem}
\label{th:erasureomni}
For $q<p$, the deterministic coding capacity of the binary erasure-erasure
adversarial channel $CE(q,p)$ is the same as that of the binary
erasing adversarial channel with an omniscient adversary.
\end{theorem}

\begin{theorem}
For $q>p$, the binary erasure-erasure myopic channel $CE(q,p)$ has secrecy 
capacity $q-p$.
\label{th:erasuresecrecy}
\end{theorem}

\subsection{More general binary input channels 
}
\label{sec:eraseflip}

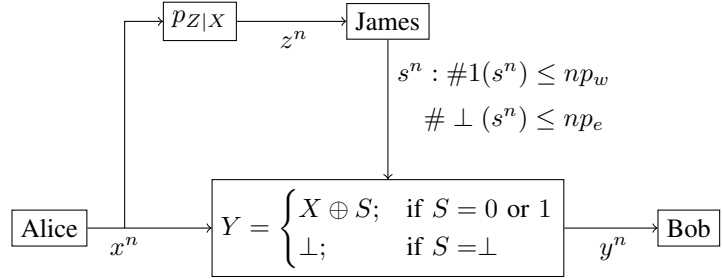
\begin{figure}[h]
\setlength{\unitlength}{1cm}
\centering
\begin{tikzpicture}[trim left, scale=1]

\node (alice) [draw, rectangle] at (0.5,3.25) {Alice};
\node (bob) [draw, rectangle] at (9.0,3.25) {Bob};
\node (bobch) [draw, rectangle] at (5, 3.25) {$Y=\begin{cases}
X\oplus S; & \text{if } S=0\text{ or }1\\
\perp; & \text{if }S=\perp\end{cases}$};
\node (calch) [draw, rectangle] at (2.5,6) {$p_{Z|X}$};
\node (James) [draw, rectangle] at (5,6) {James};

\draw [->] (alice) -- (1.5,3.25) node[below] (tx) {$x^n$} -- (bobch);
\draw [->] (tx) -- (1.5,6) -- (calch);
\draw [->] (calch) -- +(1.25,0) node[below] {$z^n$} -- (James);
\draw [->] (James) --
          +(0,-0.7) node[right] {$\bs: \#1(\bs) \leq np_w$} --
          +(0,-1.3) node[right] {$\quad \#\perp(\bs) \leq np_e$} -- (bobch);
\draw [->] (bobch) -- +(3.0,0) node[below] {$y^n$} -- (bob);

\end{tikzpicture}
\caption{The binary input channel with erasing and flipping adversary, $CEF(p_{Z|X}, p_e,p_w)$}
\label{fig:binaryinput}
\end{figure}

Similar to the binary erasure-erasure channel $CE(q,p)$, we may improve on Theorems~\ref{th:g_lower}, \ref{th:g_secrecy}
for other channels as well. A number of examples are given below.

\subsubsection{Erasing and flipping adversary, $CEF(p_{Z|X}, p_e,p_w)$}
Let us consider a binary input setup $CEF(p_{Z|X}, p_e,p_w)$ where
James's channel is given by $p_{Z|X}$, and James can erase upto a
fraction $p_e$ and flip upto a fraction $p_w$ of the transmitted bits
($p_e+p_w\leq 1$).  The corresponding random channel (binary symmetric error and
erasure channel) has a capacity
$(1-p_e)\left(1-H\left(\frac{p_w}{1-p_e}\right) \right)$.

\begin{theorem}
\label{th:cefcap}
For $CEF(p_{Z|X}, p_e,p_w)$, if 
$$
1-H(X|Z) < (1-p_e)\left(1-H\left(\frac{p_w}{1-p_e}\right) \right)
$$
then the capacity is 
$(1-p_e)\left(1-H\left(\frac{p_w}{1-p_e}\right) \right)$.
\end{theorem}
The proof
is similar to the general myopic results, and the key element in the
proof is outlined in Remark~\ref{rem:erasure-erasure}.

\subsubsection{Secrecy capacity for erasing and flipping adversary, $CEF(p_{Z|X}, p_e,p_w)$}
We now consider the secrecy capacity of $CEF(p_{Z|X}, p_e,p_w)$, i.e., 
when James's channel is $p_{Z|X}$, and the James can erase upto a
fraction $p_e$ and flip upto a fraction $p_w$ of the transmitted bits.
Using a randomly constructed code
to encode the message and private randomness we obtain (see Remark~\ref{rem:erasure-erasure} for details).
\begin{theorem}
\label{th:cefsec}
For the channel $CEF(p_{Z|X}, p_e,p_w)$, the secrecy rate 
$H(X|Z)+p_eH\left(\frac{p_w}{1-p_e}\right) - p_e-H\left(\frac{p_w}{1-p_e}\right)$
is achievable.
In particular, for the channel $CEF(BEC(q), p_e,p_w)$, the secrecy rate 
$q+p_eH\left(\frac{p_w}{1-p_e}\right) -
p_e-H\left(\frac{p_w}{1-p_e}\right)$ 
is achievable.
\end{theorem}

\subsubsection{Wiretap channel of type II with erasing and flipping adversary, $\wcefII(p_r,p_e,p_w)$}
We denote the wiretap channel of type II with active adversary who can
erase and flip bits as $\wcefII(p_r,p_e,p_w)$.
Here instead of James's channel being a random erasure channel, 
James can also choose a $p_r=1-q$ fraction of the transmitted
bits to view/read, and he can erase upto a fraction $p_e$ and
flip upto a fraction $p_w$ of the transmitted bits. 
This is a
generalization of the models studied in 
\cite{aggarwal2009wiretap, wang_on_capacity_2016} for Wiretap channel of type
II with active adversaries.

\begin{theorem}[Wiretap channel of type II with erasing and flipping adversary]
\label{th:wiretaptwo}
For $\wcefII(p_r,p_e,p_w)$, the rate
$1-p_r+p_eH\left(\frac{p_w}{1-p_e}\right) -
p_e-H\left(\frac{p_w}{1-p_e}\right)$ is achievable.
\end{theorem}
This result is of independent interest as it generalizes
results on the wiretap channel of type II with active adversary~\cite{aggarwal2009wiretap,wang_on_capacity_2016}. Our general proof technique together with
the argument in Remark~\ref{rem:erasure-erasure} implies this result.
As a special case, in the model $\wcefII(p,p,0)$, 
James can observe upto a $p$ fraction of bits of {\em
his choice} and he can also erase a $p$ fraction of bits. For a more
restricted James (who has to erase the same bits that he observes),
Aggarwal et al.~\cite{aggarwal2009wiretap} showed that rates upto $1-p -
H(p)$ can be achieved. Theorem~\ref{th:wiretaptwo} improves this to
$1-2p$.
As another special case, in the model $\wcefII(p_r,0,p_w)$,
Theorem~\ref{th:wiretaptwo} gives an achievable rate of 
$1-p_r-H(p_w)$, which is same as the achievable rate in
\cite{wang_on_capacity_2016}.

\section{Proof of Theorem~\ref{th:binarycapacity}}
\label{sec:bin}

The converse follows using the converse for \bsc (as discussed in Section~\ref{sec:results}). We first present
a sketch of the achievability in Subsection~\ref{sec:Int1}, and then
in Subsection~\ref{sec:binproof} we give a formal proof.
Our proof is summarized in Figure~\ref{fig:intuition}.

\subsection{Proof sketch for Theorem~\ref{th:binarycapacity}}
\label{sec:Int1}

We now sketch the proof for the achievability of rates arbitrarily close to $1-H(p)$ over $C(q,p)$ when $q>p$ (Theorem~\ref{th:binarycapacity}).
For a precise proof, and also for the precise definition of some of the
ideas, we refer the reader to the technical proof appearing in 
Subsection~\ref{sec:binproof}.

Our code construction uses a uniformly chosen codebook over $\{0,1\}^n$ and our decoder associates with each received word $\bY$ its closest codeword $\bX$. 
We now show that with high probability over code design, such codes are sufficient for communication at any rate $R=1-H(p)-\epsilon$ and average error $\epsilon$ where $\epsilon>0$ is arbitrarily small.
In what follows, many of the statements we make occur with high probability over code design (and not necessarily with probability 1) even though at times we do not state so explicitly.

Consider a codeword $\bX$ transmitted by Alice.
This codeword passes through the channel to James, and James receives the corrupted version $\bZ$ of $\bX$.
The fact that $q>p$ (or more precisely that $1-H(q)<1-H(p)$) now implies that there are approximately $2^{n(H(q)-H(p))}$ codewords that are {\em consistent} with James's view $\bZ$ (Lemma~\ref{lemma:shell}).
Namely, that there are an exponential number of codewords $X'^\bl$, that from James's perspective, may have been transmitted by Alice that would have resulted in the same $\bZ$.
These consistent codewords (depicted in Figure~\ref{fig:code_decomp}) are exactly those that lie in a ball around $\bZ$ of radius $\approx qn$.

\begin{figure}[t!]
\begin{center}
\includegraphics[scale=0.48]{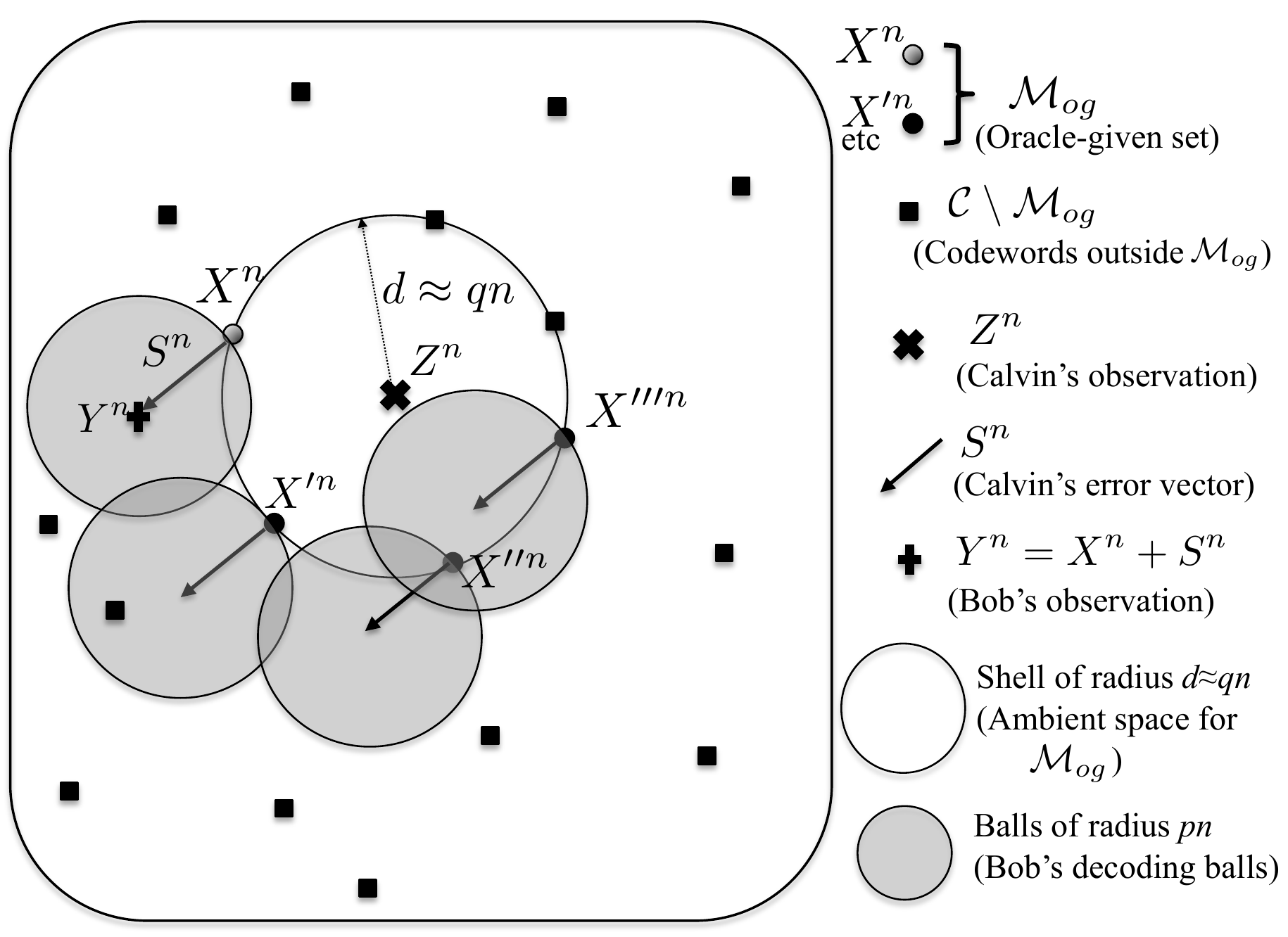}
\caption{\scriptsize{{\bf{Relationship between important notation:}} Alice
transmits $\bX$, James observes $\bZ$. The oracle informs James that Alice's
transmission is one of $\ogs = \{\bX, {X'}^\bl, {X''}^\bl, {X'''}^\bl$\}, each
of which is $d\approx q\bl$ away from $\bZ$. 
James imposes the error vector $\bS$, and Bob
receives $\bY = \bX + \bS$. Bob decodes to the hopefully unique codeword in a
ball of radius $p\bl$ around $\bY$. In this example, the ball around $X'^\bl +
\bS$ contains a codeword from $\code \setminus \ogs$, and the ball around
$X'''^\bl + \bS$ contains a codewords from $\ogs$, and therefore if James
imposes the error-vector $\bS$ an error may result in these cases. In our
proof, we show that for {\it any} $\bS$ the fraction of $\bX \in \ogs$ whose
decoding balls contain {\it any} other codewords is ``very small''.}}
\label{fig:code_decomp}
\vspace{-6mm}
\end{center}
\end{figure}

James doesn't know what $\bX$ is and he even doesn't know what $d=\|\bX-\bZ\|$
is exactly (other than the fact that $d \approx qn$), however in our analysis we
slightly help James by revealing $d$ and by choosing a small yet exponential
set of codewords (of size $2^{\bl \delta}$ for an appropriately small $\delta$)
from Alice's codebook $\code$ chosen from the set of all codewords at distance
$d$ from $\bZ$, with the additional guarantee that it also contains the true
$\bX$.  We refer to this set as the {\it oracle-given set} $\ogs$.  The
advantage of defining such a set is purely for ease of analysis using
a two-stage counting as elaborated later. 
Revealing this information to James only makes him stronger and
thus a coding scheme that succeeds here will also succeed against the original
James.
Conditioned on James's view $\bZ$ and $\ogs$, each of the codewords 
in $\ogs$ is equally likely to
have been transmitted. 
Nonetheless, James still has a reasonable amount of
uncertainty about which of the $2^{\bl \delta}$ codewords was actually
transmitted -- indeed, this is the uncertainty that we leverage in our
analysis.  

We now show that, no matter what action James takes, Bob will be able to correctly decode with high probability over the transmitted codeword $\bX$ of Alice.
First consider any codeword $\bX$ in the set $\ogs$ which is consistent with James's view $\bZ$.
A decoding error  occurs for $\bX$ if after James chooses the state vector $\bS$ it holds that $\bY=\bX +\bS$ is closer to a codeword $\hat{X}^\bl \in \code$ than it is to the transmitted word $\bX$.
In such a case we say that $\bX$ is {\em confusable} with $\hat{X}^\bl$.
In Figure~\ref{fig:code_decomp} this is expressed by the shaded region around $\bY=\bX +\bS$, if it is empty then decoding is successful and if it includes a $\hat{X}^\bl$ then a decoding error may occur.
Conditioned on James's view $\bZ$, we show that no matter what action James takes, for {\em most} codewords in $\ogs$ there will not be a decoding error.
We stress that, indeed, for every action $\bS$ there may be some codewords in $\ogs$ that will have a decoding error, but for the vast majority of them their corresponding shaded region in Figure~\ref{fig:code_decomp} will be empty. This will imply successful decoding with high probability.

We now fix any specific action $\bS$ for James and show that only a small fraction of codewords in $\ogs$ will be corrupted by $\bS$, or equivalently, only a small fraction of codewords $\bX$ will be confusable with some 
$\hat{X}^\bl \in \code$.
We consider two cases: the case that $\hat{X}^\bl \in \code \setminus \ogs$
(Lemma~\ref{lem:list1}) and the case in which $\hat{X}^\bl \in \ogs$
(Lemma~\ref{lem:list2}). 
Roughly speaking, using a careful analysis based on the Principle of Deferred Decisions \cite{alon2004probabilistic} in this case we may assume that the codewords in $\code \setminus \ogs$ are independent of those in $\ogs$. 
This allows us to bound the number of codewords in $\ogs$ that are confusable with codewords in $\code \setminus \ogs$ using certain list decoding arguments and in particular using a novel ``double-list-decoding argument''.

Recall that $\bX$ is confusable with $\hat{X}^\bl$ iff the latter lies in the ball of radius $p\bl$ around $\bX+\bS$.
Let $\Lambda$ be the union of all such balls for $\bX \in \ogs$.
Since $\code \setminus \ogs$ is independent of $\ogs$, it follows from standard list decoding arguments (and our setting of parameters) that the number of codewords $\hat{X}^\bl \in \code \setminus \ogs$ that lie in $\Lambda$ is {\em small}, specifically of size at most some polynomial in $\bl$. 
This is a first step towards our proof, which shows that the number of $\hat{X}^\bl$ that may confuse some $\bX\in\ogs$ is small.
However, each such $\hat{X}^\bl$ potentially may confuse many $\bX \in \ogs$.
To bound the number of $\bX$ that may be confused by a single $\hat{X}^\bl$, we use a second list-decoding argument.
As before, we use the independence between $\code \setminus \ogs$ and $\ogs$ to bound this number by a polynomial in $\bl$.
All in all, we conclude that at most a polynomial number of codewords $\bX$ in $\ogs$ are confusable with a codeword $\hat{X}^\bl$ in $\code \setminus \ogs$, which concludes the first case we considered.

For the second case, we bound the number of codewords $\bX \in \ogs$ that are confusable with $\hat{X}^\bl$ which are also in $\ogs$.
In this case, we may no longer na\"ively rely on the list decoding arguments used previously.
This is due to the fact that the previous arguments were strongly based on the independence between  codewords that were potentially being confused (denoted by $\bX$) and those that were confusing them (denoted by $\hat{X}^\bl$).
To overcome this difficulty, we partition the set $\ogs$ into disjoint sets and study the effect of one set in the partition on another.
Then by a combination of similar list decoding arguments and additional
counting arguments,
we can show that only an exponentially small fraction of codewords $\bX \in \ogs$ are confused by a codeword $\hat{X}^\bl \in \code$.
The claimed list decoding properties above hold with extremely high probability of $1-2^{-\beta\bl^2}$ on code design. This allows us to use the union bound on several assumptions made throughout the discussion (e.g., the values of $d$, $\bz$ and $\bs$).

\subsection{Achievability proof of Theorem~\ref{th:binarycapacity}}
\label{sec:binproof}
\input{proofs_bsc_v3.tex}

\begin{figure*}[t]
\begin{center}
\includegraphics[scale=0.74]{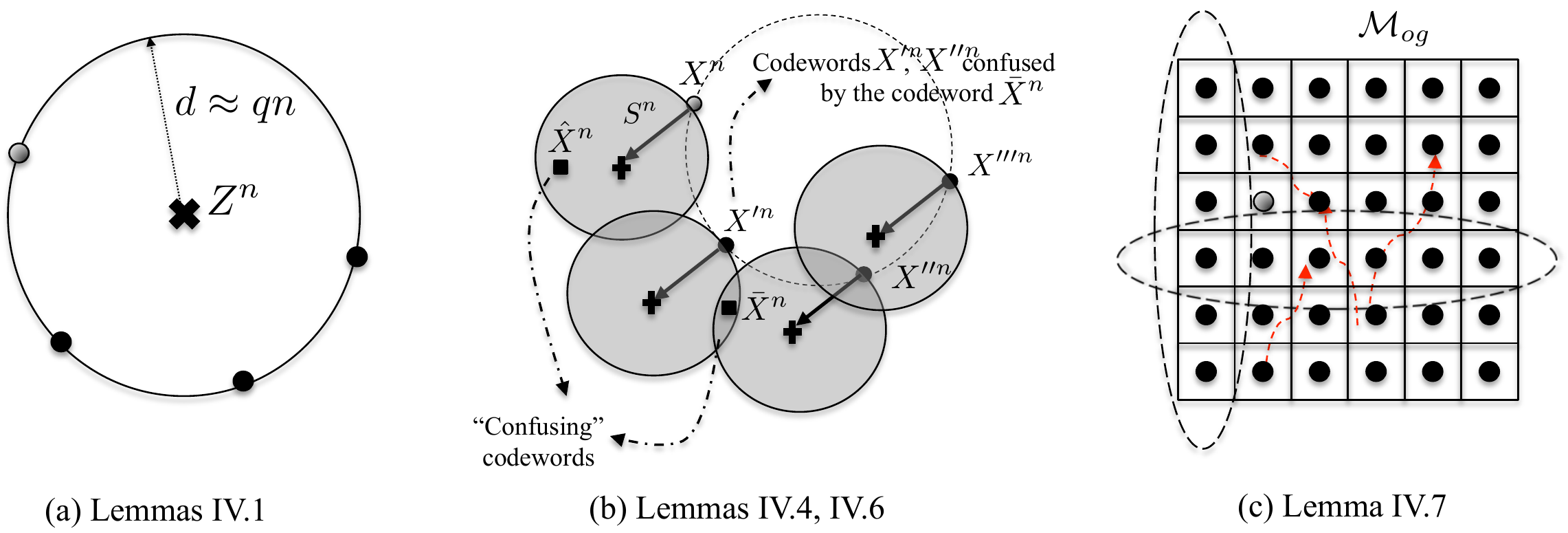}
\caption{\scriptsize{ {\bf{Intuition about proof techniques:}} {\it \underline{Fig (a)}:} Lemma~\ref{lemma:S0} uses ``standard'' concentration inequalities to argue that the value of $d$ (the amount of noise James sees) is ``close'' to $q\bl$. Lemma~\ref{lemma:shell} argues that for {\it every} shell with such $d$, the number of codewords on a shell of radius $d$ centred at $\bZ$ is close to $2^{n(H(q)-H(p))}$ -- from James's perspective, each of these codewords is equally likely to have been transmitted. {\it \underline{Fig (b)}:} Lemmas~\ref{lem:basic1} and~\ref{lem:list1} are ``list-decoding'' lemmas. Lemma~\ref{lem:basic1} argues that regardless of the shape of the volume being considered, as long as it is smaller than the ``average volume per codeword'', for a super-exponentially large fraction of codes the number of codewords in the volume is not large (at most ${\cal O}(\bl^2)$). Lemma~\ref{lem:list1} then uses this result to show that there are not too many codewords in $\ogs$ ``confused by $\bS$ with codewords in $\code\setminus \ogs$'' (at most ${\cal O}(\bl^4)$). It does this in two steps -- it first uses Lemma~\ref{lem:basic1} to show that there are not too many ``confuse{\it-ing}'' codewords from $\code\setminus \ogs$ (\eg~$\hat{X}^\bl$ and $\bar{X}^n$ in the figure), and then it re-uses Lemma~\ref{lem:basic1} to show that each ``confuse{\it-ing}'' codeword does not lead to too many ``confus{\it-ed}'' codewords ($\bar{X}^\bl$ only confuses $X'^\bl$ and $X''^\bl$ in the figure). {\it \underline{Fig (c)}:} Lemma~\ref{lem:list2} analogously proves that there are not too many codewords in $\ogs$ ``confused by $\bS$ with other codewords in $\ogs$''. To do so, the $2^{\bl \delta}$ codewords in $\ogs$ are arranged in a square grid. Using Lemma~\ref{lem:basic1} one can show that in any column (respectively row) of this grid there are not too many (at most ${\cal O}(\bl^2)$) codewords in that column (respectively row) that are confused by $\bS$ with any other codeword in any other column (respectively row) -- the red-arrows in the figure indicate codewords that are confused by $\bS$ with another codeword in a different row or column. This allows one to demonstrate that the total fraction of codewords in $\ogs$ that are confused by $\bS$ is an exponentially small fraction, and hence the probability of error is small. Since the preceding statements are true with probability super-exponentially close to $1$, one may take a union bound over all possible values of $d$, $\bZ$, and $\bS$.}}
\label{fig:intuition}
\end{center}
\end{figure*}

\section{General myopic channel and the proof of Theorem~\ref{th:g_lower}}
\label{app:gen}

\subsection{Intuition for general myopic jamming channels}

For general pairs of channels (a stochastic channel $\pee_{Z|X}$ and the AVC $\pee_{Y|XS}$), many of the ideas used in the achievability proofs for the bit-flipping channel also go through to result in the Theorems~\ref{th:g_lower}, \ref{th:g_upper}, \ref{th:g_secrecy}. We highlight the major differences here.

{\bf \underline{Model:}}
To begin with, we borrow heavily from the problem formulation in~\cite{sarwate_coding_2010} to model the relationship between the {\it $\bl$-letter tuple} $(\bX,\bZ,\bS,\bY)$ in terms of {\it single-letter distributions} $\pee_{X}\pee_{Z|X}\pee_{S|Z}\pee_{Y|XS}$. More precisely: 
\begin{itemize}
\item In our achievability proofs we focus on codebooks generated i.i.d. according to the distribution $\pee_X$ (subject to the constraint $\pee_X \in \cV$).
\item The channels $\pee_{Z|X}$ and $\pee_{Y|XS}$ are specified as part of the problem statement, in which they are assumed to be memoryless.
\item While James is free to choose any $\bl$-letter state-vector (subject to the constraint that the type-class of $\bS$ be in $\cW$) as a function of his full $\bl$-letter observation $\bZ$, any such choice can be ``reverse engineered'' as a particular length-$\bl$ instantiation of a conditional distribution $\pee_{S|Z}$, where the $\pee_{S|Z}$ corresponds to the conditional type-class of $\bS$ given the observed vector $\bZ$.
\item Further, as defined in the myopic jamming problem statement, James's state vector $\bS$ must be conditionally independent of $\bX$ given his observation $\bZ$.
\end{itemize}
Putting these together, we note that the {\it $\bl$-letter tuple} $(\bX,\bZ,\bS,\bY)$ can be thought of as a tuple of length-$\bl$ sequences generated according to the single-letter distribution $\pee_{X}\pee_{Z|X}\pee_{S|X}\pee_{Y|{XS}}$ (subject, of course, to the corresponding input and state constraints $\cV$ and $\cW$).

{\bf \underline{Results:}}
It is then natural to conjecture that the corresponding ``capacity'' of the
problem equals $I(X;Y)$, maximized over all valid encoder profiles
corresponding to $p_X$, and minimized over all valid attacks by James
corresponding to $p_{S|Z}$, as long as the channel $p_{Z|X}$ is ``sufficiently
myopic'' compared to the channel $p_{Y|XS}$. And indeed both our achievability
and converse expressions (in Theorems~\ref{th:g_lower}, \ref{th:g_upper}) take
this form. However, there is a gap between the tightest converse we can prove
in Theorem~\ref{th:g_upper}, and the achievability we can prove in
Theorem~\ref{th:g_lower}. Specifically, our achievability requires us to
maximize over $p_X$ such that two ``myopicity'' conditions (given by (b) and
(c) in the statement of Theorem~\ref{th:g_lower}) are satisfied.

For a wide variety of other myopic channel-pairs our approach results in non-trivial achievability results. This includes problems in which the channels from Alice to James, and from Alice to Bob, are of different ``forms'' (for instance, a \bscq\ from Alice to James, and an AVC from Alice to Bob in which James can {\it erase} a fraction $p$ of bits).
However, we are by no means convinced that the achievability result we present is optimal in general. 
In particular, while the first myopicity condition (Theorem~\ref{th:g_lower} (b)), $I(X;Z) < I(X;Y)$ is ``somewhat natural'' (corresponding to James having a weaker channel than any channel he can impose on Bob), the second condition (Theorem~\ref{th:g_lower} (c)) arises from somewhat technical considerations in ``reverse list-decoding'' described below.

{\bf \underline{Proof Techniques:}}
Given the problem formulation described above, many (but not all) of the ideas described in the achievability proof of the channel $C(q,p)$ carry through. Specifically:
\begin{itemize}
\item The oracle-given set is defined in an analogous manner to how it is
defined for the channel $C(q,p)$. Specifically, as long as $I(X;Z) < I(X;Y)$,
one can demonstrate via standard combinatorial arguments that all ``typical
type-classes'' have exponentially many codewords in them -- one then constructs
the corresponding oracle-given set by choosing sufficiently many
codewords with the same empirical conditional distribution $p_{X|Z}$ as the
``true'' $(\bX,\bZ)$.  \item One can then show that over the randomness in
which codeword $\bX$ in the oracle-given set was actually transmitted, with
high probability the tuples $(\bX,\bS,\bY)$ are jointly typical according to
the joint distribution $\pee_{X}\pee_{Z|X}\pee_{S|X}\pee_{Y|{XS}}$.
\item For state-deterministic channels~\footnote{The restriction to state-deterministic channels is a technical condition required by our proof, so as to be able to achieve the claimed rate. Removing this restriction is possible, but then with our current proof techniques our achievable rates reduce to $I(X;Y) - H(Y|X,S)$ rather than $I(X;Y)$, and also we need to further restrict the class of state constraints $\cW$ for which such a result is possible. Attempting to remove this restriction is a direction of ongoing research. Nonetheless, the class of state-deterministic channels already contains many interesting myopic channels problems, such as the channel $C(q,p)$ discussed in length above, the channel $CE(q,p)$ discussed in Sec.~\ref{sec:erasure}, and a variety of other ``mixed'' channels. \label{foo:prob}} $\pee_{Y|{XS}}$ one can show an ``$\bX$''-list-decoding argument. Namely, one can show that with probability super-exponentially close to one over code design, the number of codewords from the oracle-given set that are translated under the action of a feasible state vector $\bS$ into the conditionally typical set w.r.t.  any typical $\bY$ is ``small'' (at most ${\cal O}(\bl^2)$).
\item We also use a ``$\bY$''-list-decoding argument, that demonstrates that the number of $\bY$ resulting from the action by a fixed feasible state vector $\bS$ acting on any $\bX$ from the oracle-given set is at most ${\cal O}(\bl^2)$. Again, for technical reasons, to prove this list-decoding result we need to impose some additional constraints on the class of permissible states $\cW$ -- these are precisely the constraints in Theorem~\ref{th:g_lower}(c). 
\item Finally, the square-grid argument described in the proof of the $C(q,p)$ is essentially unchanged in this general setting.
\item As in the $C(q,p)$ case, the interplay between the problems of jamming-resilient and eavesdropping-resilient code-designs is natural. As in the general wiretap-channel case, one can indeed use our codes, as outlined in Theorem~\ref{th:g_lower}, and transmit over them messages that themselves comprises of the ``true message'', padded with random bits, and encoded via wiretap-channel codes, so as to guarantee communication that is both reliable against James's jamming, and secure against his eavesdropping.
\end{itemize}

\subsection{Proof of Theorems~\ref{th:g_lower}}
\input{proofs_gen4.tex}

\section{Conclusions}
In this work we study the secure and standard capacity of
adversarial myopic channels.  For the bit-flipping adversarial channel
$C(q,p)$, for the binary erasure-erasure adversarial channel $CE(q,p)$, 
and more generally for binary input channels where the adversary can 
both erase and flip some fractions of bits, we characterize these capacities as the capacity under random noise
when the adversary's own channel is more noisy than the worst noise it
can impose on Bob, in terms of mutual information. 
For these models, we also consider analogs of the wiretap channel of type II.
For general myopic channels, we prove similar achievability results
under a stricter condition of myopicity. 
A tight
characterization of capacity for general myopic channels is left open and
subject of future work.

\section{Acknowledgement}
The work of B.~K.~Dey was supported in part by Bharti Centre
for Communication, IIT Bombay, a grant from the Department
of Science and Technology, Government of India, and a grant from
ITRA,  Government of India.
S.~Jaggi's work was partially supported by a grant from University Grants Committee of the Hong Kong, Special Administrative Region, China (Project No. AoE/E-02/08), and by the Innovation Technology Fund (ITF) Project ITS/143/14FP, Innovation Technology Commission, the Government of HKSAR.
The work of M.~Langberg was supported in part by NSF grant no. 1321129.

{\scriptsize{
\bibliographystyle{unsrt}
\bibliography{./Myopic.bib}
}}

\end{document}

%% file: table.tex
{\tiny{
\begin{table}
\begin{tabular}{|C{.5cm}|L{2cm}|C{1.2cm}|C{1.2cm}|C{1.0cm}|C{1.0cm}|L{2.3cm}|L{2.7cm}|C{2cm}|}
\hline\hline
&& {\bf Channel: Alice \ \ \ \ to \ \ \ \ \ James} & {\bf Channel: Alice \ \ \ \ to \ \ \ \ \ Bob} 
& {\bf  \begin{turn}{-90}available?\end{turn}\  \begin{turn}{-90}randomness\end{turn}\ \begin{turn}{-90}Common\end{turn}} & 
{\bf \begin{turn}{-90}coordinate?\end{turn}\  \begin{turn}{-90}and Jammer\end{turn}\ \begin{turn}{-90}Eavesdropper\end{turn}}
&
{\bf Reliable throughput} & {\bf Reliable + secure throughput} & {\bf Comments}\\ \hline \hline
\multirow{ 5}{0cm}{\bf \begin{turn}{-90}Binary input channels\end{turn}} & \multirow{ 2}{*}{\cite{shannon_mathematical_1949} Sha49} & -- & BSC(p) & -- & -- & $1-H(p)$ & -- & \multirow{ 2}{2cm}{\it Baseline random noise channels}\\ \cline{3-8}
&& -- & $pn$ erasures & -- & -- & $1-p$ & -- & \\\cline{2-9}
& \multirow{ 2}{2cm}{\cite{gilbert_comparison_1952} Gil52, \cite{varshamov_estimate_1957} Var57, \cite{mceliece_new_1977} McERRW77} & Perfect channel &
$pn$ bit-flips & No & Yes & $\geq 1-H(2p),$ $\leq LP(2p)$ & -- & 
\multirow{ 2}{2cm}{\it Baseline adversarial noise channels}\\\cline{3-8}
&& Perfect channel &
BEC(p) & No & Yes & $\geq 1-H(p),$ $\leq LP(p)$ & -- &\\ \cline{2-9}
& \cite{csiszar1988capacity} CsiN88 \cite{Langberg:08} Lan08 & No channel
& $pn$ bit-flips & No & -- & $1-H(p)$ & -- & {\it ``Oblivious'' jammers}\\ \cline{2-9}
& \cite{aggarwal2009wiretap} AggLCP09 & \multicolumn{2}{|L{2.4cm}|}{James can observe and erase any $pn$ bits of his choice} & No & Yes & -- & $1-p-H(p)$ & {\it Wiretap type II with erasing adversary}\\\cline{2-9}
& \cite{wang_on_capacity_2016} Wan16  & James chooses $p_rn$ bits to see & $p_wn$ bit flips & No & Yes & -- & $1-p_r-H(p_w)$ & {\it Wiretap type II with flipping adversary}\\\hline \hline
\multirow{4}{0cm}{\bf \begin{turn}{-90}Our results: binary input\end{turn}} 
& $C(q,p)$, $q>p$ & BSC(q) & $pn$ bit flips & No & Yes & $1-H(p)$ & $H(q)-H(p)$ & {\it Th.}~\ref{th:binarycapacity},~\ref{th:binarysecrecy} \\\cline{2-9}
& $CE(q,p)$, $q>p$ & BEC(q) & $pn$ erasures & No & Yes & $1-p$ & $q-p$ &
{\it Th.}~\ref{th:erasurecapacity},~\ref{th:erasuresecrecy}\\\cline{2-9}
& CEF ($p_{Z|X}, p_e, p_w$) & $p_{Z|X}$ & $p_en$ erasures \& $p_w$ flips
& No & Yes & $(1-p_e)$ $\left(1-H\left(\frac{p_w}{1-p_e}\right)\right)^\dagger$ &
$H(X|Z)-p_e$ $+p_eH\left(\frac{p_w}{1-p_e}\right)$ $-H\left(\frac{p_w}{1-p_e}\right)$ & {\it Th.}~\ref{th:cefcap},~\ref{th:cefsec}\\\cline{2-9}
& WCEF-II ($p_r, p_e, p_w$) & James can choose $np_r$ bits to see &
$p_en$ erasures, $p_wn$ flips & No & Yes & -- & $1-p_r-p_e$ $+p_eH\left(\frac{p_w}{1-p_e}\right)$ $-H\left(\frac{p_w}{1-p_e}\right)$ & {\it Wiretap type II with erasing and flipping adv. Th.}~\ref{th:wiretaptwo}\\\hline \hline
\multirow{2}{0.8cm}{\bf \begin{turn}{-90}(large alphabets)\end{turn} \begin{turn}{-90}Other channels\end{turn}} & \cite{safavi2015secure,wang2016model}  & James can choose $p_rn$ symbols to see & Additive noise over ${\mathbb F}_q$, $p_wn$ symbol errors & No & Yes & $1-p_w$ & $1-p_r-p_w$ & {\it $q$ ``sufficiently large'', Computationally efficient}\\\cline{2-9}
& \cite{zhang2015coding,kadhe2015reliable,zhang2015talking}  & James can observe $n(p_{ro}+p_{rw})$ symbols of his choice & Additive and overwrite noise over ${\mathbb F}_q$, $n(p_{wo}+p_{rw})$ symbol errors & No & Yes & $f_1(p_{ro},p_{rw},p_{wo})$ & $f_2(p_{ro},p_{rw},p_{wo})$ & {\it $q$ ``sufficiently large'', Computationally efficient, $f_1(.),f_2(.)$ -- complete characterization}\\
\hline \hline
\multirow{ 3}{.8cm}{\bf \begin{turn}{-90}Memoryless Channels\end{turn}\ \begin{turn}{-90}General Discrete\end{turn}} & \cite{csiszar1988capacity} CsiN88 & No channel & General $p_{Y|X,S}$ & No & -- & $\max_{p(x)}\min_{p(s)}$ $I(X;Y)$ & -- & {\it Max/min over permissible input/state constraints resp.}\\\cline{2-9}
& \cite{molavianjazi2009arbitrary} MolBL09, \ \ \ \cite{liang2009compound} LiaKPS09 & $p_{Z|X,S}$ & $p_{Y|X,S}$ & Yes & No & -- & $\geq \max_{p(x,u)}$   $( \min_{s} I(U;Y) - \max_{s} I(U;Z) )$
$\leq\max_{p_X} \min_{p_{S}}$ $I(X;Y|Z)$ & {\it $\forall s, U\leftrightarrow X \leftrightarrow (Y,Z)$, ``Degraded'' wiretapper }\\\cline{2-9}
& \cite{sarwate_coding_2010} Sar10 & $p_{Z|X}$ & $p_{Y|X,S}$ & Yes & Yes &
$\max_{p_X} \min_{p_{S|Z}}$ $I(X;Y)$ & -- & \\
%
\cline{2-9}
& $C(p_{Z|X}, \cW)$ \ (our result) & $p_{Z|X}$ & state constraint: $\cW$ &
No & Yes & $\max_{p_x} \min_{p_{S|Z}}$ $I(X;Y)^\ddagger$ & $\max_{p_x} \min_{p_{S|Z}}$
$(I(X;Y) - I(X;Z))^\ast$ & Th.~\ref{th:g_lower},~\ref{th:g_secrecy}\\
\hline \hline
\end{tabular}
\vspace{3mm}
\caption{Summary of models and results. ($\dagger$) under the condition $1-H(X|Z) < (1-p_e)\left(1-H\left(\frac{p_w}{1-p_e}\right) \right)$; ($\ddagger$) under the conditions: $\forall p_{S|Z} \in \cW_{S|Z}:\ I(X;Z) < I(X;Y)$, and
$\max_{p_{S|Z} \in \cW_{S|Z}}H(X|Y,S) + \max_{p_{S|Z} \in 
\cW_{S|Z}}H(Y|X) < H(X|Z)$; ($\ast$) under the condition 
$\max_{p_{S|Z} \in \cW_{S|Z}}H(X|Y,S) + \max_{p_{S|Z} \in 
\cW_{S|Z}}H(Y|X) < H(X|Z)$.}
\label{table}
\end{table}
}}

%% file: proofs_bsc_v3.tex
 \def\cSh{{\mathcal Sh}}
 \def\hone{{h_1(\etwo,\eone)}}
 \def\htwo{{h_2(\etwo,\eone)}}
 \def\hthree{{h_3(\etwo,\eone)}}
 \def\hfour{{(3\delta /4)}}
 \def\hfive{{h_1(\etwo)}}
 \def\mosh{{\mathcal M}_s}
 \def\ogm{{\mathcal M}_{og}}
 \def\omosh{{\mathcal M}_o}

\newcommand{\oogs}[2]{{\cC_0(#1,#2)}}
\newcommand{\aau}[2]{{\cC_1(#1,#2)}}
\newcommand{\uuni}[2]{{\cC_2(#1,#2)}}

\newcommand{\bOGS}{{\cC_0(\bZ, r)}}
\newcommand{\bogs}{{\cC_0(\bz,r)}}
\newcommand{\bAU}{{\cC_1(\bZ,r)}}
\newcommand{\bau}{{\cC_1(\bz,r)}}
\newcommand{\bUNI}{{\cC_2(\bZ,r)}}
\newcommand{\buni}{{\cC_2(\bz,r)}}
\newcommand{\Scal}{{\cSh(\bZ,D)}}

We now prove that any rate $R< 1-H(p)$ can be achieved. Without loss of
generality, we assume that $R=1-H(p)-\epsilon > 1- H(q)$.
We also assume, without loss of generality, that $\epsilon$ is sufficiently 
small.

{\it Code construction:} 
The code consists of $2^{nR}$ vectors
$\bX (w); w =1,2,\cdots,2^{nR}$, all selected independently with i.i.d. 
$\sim Bernoulli(1/2)$ components. 

{\it Encoding:} Alice encodes message $W$ with $\bX(W)$ and transmits.

{\it Decoding:} Let the vector received by Bob be $Y^n$. If Bob finds a
unique $\hat{W}$ such that $X^n(\hat{W})$ is within distance 
$(p+\etwo)n$ from $Y^n$, 
then he declares $\hat{W}$ as the decoded message. Otherwise he declares
error. Here $\etwo$ is a predetermined constant (that is set to be sufficiently small).

For any subset of messages $M\subseteq \{1,2,\cdots, 2^{nR}\}$,
we define its codewords as 
$$X^n(M) := \{X^n(w)| w\in M\}.$$

After observing $\bZ$, James can find all the codewords which
are jointly typical with it. With high probability, the transmitted
codeword belongs to that set. 
We define, for any $\bz, d$, the ball and the shell
\begin{align*}
\cB(\bz, d) : = \{\bx\in \{0,1\}^n: d_H (\bx, \bz)) \leq d))\},\\
\cSh(\bz, d) : = \{\bx\in \{0,1\}^n: d_H (\bx, \bz)) = d))\}.
\end{align*}
Let $D$ denote the random variable $d_H (\bX(W), \bZ)$.
Let $\eone$ be a sufficiently small constant.
We define the following events
\begin{align*}
\evcf (\bz,d) :=  &\left\{ 2^{n(H(d/n)+R-1+\eone)} \geq |\{w: X^n(w)\in \cSh(\bz,d)\}| \right.
\left. \geq 2^{n(H(d/n)+R-1-\eone)} \right \}\\
\evcf := & \cap_{\bz, d:|d-qn|\leq\etwo n} \evcf (\bz,d)\\
\evd :=  &\{(q-\etwo)n \leq D\leq (q+\etwo)n\},
\end{align*}
$\Scal$ is the spherical shell around James's received vector where 
the transmitted codeword lies. The code satisfies
$\evcf$ with high probability according to Lemma~\ref{lemma:shell} below. 
We assume that an oracle reveals to James some additional information, and
prove our achievability under this stronger adversary.
For every possible $\bz$, and $d$ satisfying 
$(q-\etwo)n \leq d \leq (q+\etwo)n$, the oracle partitions the set 
of messages 
$$ M_s(\bz, d) = \{w: X^n(w)\in \cSh(\bz,d)\} $$
with codewords on the shell $\cSh(\bz,d)$ into disjoint 
subsets $M_{og}^{(1)}(\bz, d), M_{og}^{(2)}(\bz, d), \cdots, M_{og}^{(\lambda(\bz, d))}(\bz, d)$
of size $2^{n\delta}$, except possibly the last subset with smaller size.
Here $\delta$ is chosen to be small enough to satisfy some requirements
to be mentioned later.
This partitioning is done deterministically by taking the messages in order
of their value, that is, satisfying $w<w'$ for each 
$w\in M_{og}^{(i)}(\bz, d), w' \in M_{og}^{(i+1)}(\bz, d)$.

{\it Additional information to James from the oracle:} The oracle reveals $D$ and the particular subset 
$\ogm = M_{og}^{(i)}(\bZ, D)$ of $M_s(\bZ, D)$ that contains the 
encoded message $W$. 

We denote the sets
\begin{align*}
&\mosh := M_s(\bZ, D)\\
&\omosh:= \mosh \setminus \ogm\\
\end{align*}
We make a few observations prior to our analysis.
\begin{enumerate}
\item [O.1] James's view (oracle aided) consists of $\bZ$ received over
his channel, and $\ogm, D$ received from the oracle.
\item [O.2] Given James's view, the encoded message $W$ is uniformly distributed
in $\ogm$.
\item [O.3] Over the random code construction, given the values of
$D, \bZ, \ogm, \mosh, \omosh$,
the codewords $\{\bX (w)|w\in \ogm\}$, $\{\bX (w)|w\in \omosh\}$, and
$\{\bX (w)|w\in \mosh^c\}$ are independently and uniformly distributed
in respectively $\Scal$, $\Scal$, and $\Scal^c$.
\end{enumerate}

First we give three standard lemmas showing that w.h.p., $D$ has
a typical value (Lemma~\ref{lemma:S0}), the spherical shell around $\bZ$ has a typical number
of codewords (Lemma~\ref{lemma:shell}), and that among the partitions of these codewords, the transmitted
message is not in the last one - with a smaller size than $2^{n\delta}$(Lemma~\ref{lem:E2}).

\begin{lemma}
\label{lemma:S0}
There exists $\hfive$ with $\hfive \rightarrow 0$ as $\etwo \rightarrow
0$ s.t. $Pr(\evd) \geq 1-2^{-n\hfive}$ for large enough $n$.
\end{lemma}
\begin{proof}
The proof follows from standard typicality arguments.
\end{proof}

\begin{lemma}
\label{lemma:shell}
There exists $\htwo$ s.t. $Pr (\evcf ) \geq 1- 2^{-2^{n\htwo}}$, where
$\htwo \rightarrow H(q) - H(p)$ as $\eone, \etwo \rightarrow 0$.
\end{lemma}
\begin{proof}
Note that there are at most $n+1$ possible values of $d$ satisfying the
condition in $\evcf$. Further, there are $2^n$ possible values of $\bz$.
Let us fix a pair $\bz, d$. For large enough $n$, clearly, 
$2^{n(H(d/n)+R-1+\eone /2)}\geq E|\{w: X^n(w)\in \cSh(\bz,d)\}| 
\geq 2^{n(H(d/n)+R-1- \eone /2)}$. Thus by Chernoff bound and by taking the union
bound over $\bz, d$, we have $Pr (\evcf^c) \leq 2^{-(1/6)2^{n(H(q-\etwo)+R-1- \eone)}}
\leq 2^{-2^{n\htwo}}$ for a suitable $\htwo$. 
\end{proof}

Given $\evcf, \evd$, the messages in $\mosh$ are partitioned into
exponentially many subsets. Only at most one of those subsets is of smaller
size than $2^{n\delta}$. Let us define the event
$$
\evo := \{\ogm \neq M_{og}^{(\lambda(\bZ,D))}(\bZ,D)\}
$$
that the oracle given set to James is not the last subset (which has 
possibly smaller size than $2^{n\delta}$) in the partition of $M_s(\bz,d)$.
Since given $\evcf, \evd$, $\mosh$, the encoded message
$W$ is uniformly distributed in $\mosh$, we have
\begin{lemma}
\label{lem:E2}
$$Pr(\evo|\evcf, \evd) > 1- 2^{-n\hthree}$$ where $\hthree = H(q-\etwo)+R-1-\eone-\delta$.
\end{lemma}

Let us fix $\bz, \bs$ and $d$ satisfying $|d-qn|\leq n\etwo$.
Let us now define an event $\evcs (\bz,d,\bs)$ over the code construction that
every subset of size $2^{n\delta}$ in the ordered
partition of $M_s(\bz, d)$ has at most only $2^{n\hfour}$ 
messages which will undergo decoding error for the respective realizations
of $D, \bZ, \bS, W$. We also define
$$\evcs = \cap_{\bz,d,\bs} \evcs(\bz,d,\bs),$$
where the intersection is over all $d$ satisfying $|d-qn|\leq n\etwo$, all
feasible $\bs$, and all $\bz$.

We will show in the subsequent analysis that
the random code construction guarantees
\begin{align}
& Pr(\evcs^c \cup \evcf^c) \leq 2^{-cn^2} \label{eq:toshow}
\end{align}
for a positive constant $c$. Hence, with very high probability over
the code construction, the code satisfies the good event $\evcf\cap \evcs$.
We will study such codes satisfying $\evcf,\evcs$.
For such a code, the probability of error is bounded as
\begin{align}
Pr(Error|\evcf,\evcs) & \leq Pr(\evd^c|\evcf ,\evcs) + Pr(\evo^c|\evcf ,\evd,\evcs) \notag \\
& \hspace{5mm} + Pr(Error|\evcf,\evd,\evo,\evcs)\notag \\
& \leq 2^{-n\hfive}+2^{-n\hthree}+2^{-n\delta/4} \label{eq:Perror}
\end{align}

Here the first term follows from Lemma~\ref{lemma:S0} by noting that $\evd$ is
depends on the noise realization in James's channel, and so it does not
depend on the code events $\evcf,\evcs$. The second term follows from
Lemma~\ref{lem:E2} as the same result holds even when conditioned on
$\evcs$ (Conditioned on $\evcf ,\evd$, it is independent of $\evcs$.) 
This is because, $\evo$ depends on the oracle's random choice,
and it's probability does not change even if the code satisfies the
additional property $\evcs$. Finally, the third term follows from the
definition of $\evcs$.

Hence once \eqref{eq:toshow} is proved, it will imply that with high probability,
the randomly generated code will have exponentially small probability
of error as guaranteed by \eqref{eq:Perror}. This will complete the proof
of Theorem~\ref{th:binarycapacity}.

We now proceed to prove \eqref{eq:toshow}. Now,
\begin{align*}
& Pr(\evcf^c\cup \evcs^c) = Pr(\evcf^c) + Pr(\evcf \cap \evcs^c)
\end{align*}
The first term is small by Lemma~\ref{lemma:shell}.
Now, the second term is
\begin{align*}
 Pr(\evcf \cap \evcs^c)
& =  Pr(\cup_{\bz,d,\bs} (\evcf \cap \evcs^c(\bz,d,\bs)))\\
& \leq \sum_{\bz,d,\bs} Pr(\evcf \cap \evcs^c(\bz,d,\bs))\\
& \leq \sum_{\bz,d,\bs} Pr(\evcf(\bz,d) \cap \evcs^c(\bz,d,\bs))\\
& = \sum_{\bz,d,\bs} Pr(\evcf(\bz,d)) Pr(\evcs^c(\bz,d,\bs)|
    \evcf(\bz,d))\\
& \leq \sum_{\bz,d,\bs} Pr(\evcs^c(\bz,d,\bs)| \evcf(\bz,d))\\
& = \sum_{\bz,d,\bs} \sum_M Pr(M_s(\bz,d)=M|\evcf(\bz,d))
 Pr(\evcs^c(\bz,d,\bs)| \evcf(\bz,d),M_s(\bz,d)=M)\\
& = \sum_{\bz,d,\bs} \sum_M Pr(M_s(\bz,d)=M|\evcf(\bz,d))
Pr(\evcs^c(\bz,d,\bs)| M_s(\bz,d)=M)\\
\end{align*}
where all the above summations are over $d$ satisfying $|d-qn|\leq n\etwo$,
and over $M$ satisfying $2^{n(H(d/n)+R-1+\eone)}\geq |M|\geq 2^{n(H(d/n)+R-1-\eone)}$. In the last  line, we have used the fact that $M_s(\bz,d)=M$
implies $\evcf(\bz,d)$.
Thus to show \eqref{eq:toshow}, it is sufficient to show that for some $\beta>0$, for every such $M$, 
\begin{align}
Pr(\evcs^c(\bz,d,\bs)|M_s(\bz,d)=M) \leq e^{-\beta n^2}. \label{eq:e3}
\end{align}
We note that for a given $M$, 
the partitioning is in increasing order
of the message value, and is thus deterministic. There are only 
exponentially many subsets in its ordered partition as used by the oracle.

We now proceed to prove \eqref{eq:e3}. The messages in $M$ which
contribute to $\evcs^c(\bz,d,\bs)$ are classified into two
categories. Lemma~\ref{lem:list1} bounds the number of codewords which are decoded
wrongly due to confusion with another codeword {\em outside the same
partition} (revealed by the oracle). These codewords include those
in the shell (but in another partition) as well as those outside the
shell. Lemma~\ref{lem:list2} bounds the number of codewords which
are decoded wrongly due to confusion with another codeword {\em in the
same partition} that is revealed by the oracle.

First, we give a basic list-decoding result that will be used in
Lemmas~\ref{lem:list1} and \ref{lem:list2}.

\begin{lemma}
\label{lem:basic1}
Let $A$ be a set with $2^{\alpha n}$ elements for some $\alpha >0$,
$c$ be a sufficiently large constant, $\nu>0$, and let
$X_1,X_2,\cdots, X_{N}$ be chosen uniformly at random from $A$ where
$N=2^{nR}$. If $V\subset A$ with
$|V| \leq 2^{n(\alpha-R-\nu)}$, then $Pr\{|\{i:X_i\in V\}| > cn^2\}
\leq e^{-cn^2/6}$.
\end{lemma}
The proof follows using the same argument as \cite[Lemma A.3]{Langberg:08}.

\begin{corollary}
\label{cor:list}
With probability at least $1-2^{-\beta n^2}$, the code satisfies the property
that in every Hamming sphere of radius $(p+\etwo )n$, there are at most
$cn^2$ codewords.
\end{corollary}

\begin{lemma}
\label{lem:list1}
There exists $\beta >0$ such that, for every $\bz$, $d$ satisfying 
$|d-nq|\leq n\etwo$, error $\bs $ introduced by James, for every subset 
$M$ of messages with $2^{n(H(d/n)+R-1+\eone)}\geq |M|\geq 2^{n(H(d/n)+R-1-\eone)}$,
conditioned on $M_s(\bz,d)=M$, with probability at least $1-e^{-\beta n^2}$ 
over the code, for every $i< \lambda(\bz, d)$, there are at most $c^2 n^4$ 
codewords $\bX (w)$ in $M_{og}^{(i)}(\bz,d)$ for which there is a different codeword 
$\bX \in M_s^c(\bz,d)\cup (M_s(\bz,d)\setminus M_{og}^{(i)}(\bz,d))$ with 
\begin{align}
d_H(\bX (w) + \bs, \bX) \leq (p+\etwo)n.
\label{eqn:list1}
\end{align}
\end{lemma}
\begin{proof}
We first prove the statement for codewords in $M_s^c(\bz,d)$. 
Note that the codewords in each of $M_{og}^{(i)}(\bz,d), 
(M_s(\bz,d)\setminus M_{og}^{(i)}(\bz,d))$, and $M_s^c(\bz,d)$ are 
uniformly distributed in the respective spaces. We have two key steps using 
Lemma~\ref{lem:basic1}. 

1. For every realization of $X^n(M_{og}^{(i)}(\bz,d))$, by considering 
$V=\cup_{w\in M_{og}^{(i)}(\bz,d)} \cB (\bX(w)+\bs, (p+\etwo)n)$ in Lemma~\ref{lem:basic1},
with probability at least $1-e^{-\beta n^2}$, there are at most $cn^2$
messages in $M_s^c(\bz,d)$ with codewords $X^n$ satisfying \eqref{eqn:list1} for some $w\in M_{og}^{(i)}(\bz,d)$. 
So the same statement is also true over the random choice of
$X^n(M_{og}^{(i)}(\bz,d))$.

2. Now, by Corollary~\ref{cor:list}, with probability at least
$1-e^{-\beta n^2}$, for every $\bx$, there are at most $cn^2$ codewords
$X^n(M_{og}^{(i)}(\bz,d))\cap \cB (\bx +\bs, (p+\etwo)n)$.
This ensures that there are at most $cn^2$ codewords from $X^n(M_{og}^{(i)}(\bz,d))$ 
for which $\bx \in \cB(\bX(w) +\bs, (p+\etwo)n)$.

So, for every $\bz, r$, and $\bs$, with high probability over the code,
there are at most $c^2 n^4$ codewords from $\bX (M_{og}^{(i)}(\bz,d))$
which satisfy the condition in the lemma.
Finally, taking the union bound over all $i<\lambda (\bz,d)$, 
we have the result for $M_s^c(\bz,d)$.

The same proof steps also work for 
$M_s(\bz,d)\setminus M_{og}^{(i)}(\bz,d)$.  We take
$V=\left(\cup_{w\in M_{og}^{(i)}(\bz,d)} \cB (\bX(w)+\bs, (p+\etwo)n)\right)
\cap \cSh (\bz, d)$ in the first step. We note that $|V| \leq
2^{n(H(p+\etwo)+\eone/2+\delta)}$ (for large enough $n$), $|\cSh (\bz, d)|
\geq 2^{n(H(q-\etwo)-\eone/2)}$, and the number of messages in
$M_s(\bz,d)\setminus M_{og}^{(i)}(\bz,d)$ is at most
$2^{n(H(q+\etwo)+R-1+\eone)}$. So, the expected number of these codewords
in $V$ is at most $2^{n(H(p+\etwo)+H(q+\etwo)-H(q-\etwo)+R-1+2\eone +\delta)}$.
The exponent is $<0$ for small enough $\eone, \etwo, \delta$. Thus 
the first step follows using the same arguments. The second step is the same
as before.
\end{proof}

In the following, we consider the codewords in $M_{og}^{(k)}(\bz,d)$ 
arranged in a square and indexed by 
a set $A\times A$ with $|A|= 2^{n\delta/2}$,
before being randomly associated to the messages. With abuse of notation,
for $i,j\in A$, we will denote the $(i,j)$-th codeword in this arrangement as
$\bX(i,j)$ (note that this omits the global association of the codewords
to the actual messages; to avoid this, we may further index these
with $(\bz,d,k)$.) 

\begin{lemma}
\label{lem:list2}
There exists $\beta >0$ such that, for every $\bz$, $d$ satisfying
$|d-nq|\leq n\etwo$, error $\bs $ introduced by James, for every subset
$M$ of messages with $2^{n(H(d/n)+R-1+\eone)}\geq |M|\geq 2^{n(H(d/n)+R-1-\eone)}$,
conditioned on $M_s(\bz,d)=M$, with probability at least $1-e^{-\beta n^2}$
over the code, for every $k< \lambda(\bz, d)$, (i) for every $i$, there are at most $c^2 n^4$
codewords in the $i$-th row $\{\bX (i,j)|j\in A\}$ for which
$\{\bX (i',j')|i' \neq i,j'\in A\}\cap \cB (\bX (i,j)+\bs, (p+\etwo)n)$ is non-empty, (ii)
there are at most $2^{n\hfour}$ messages $w\in
M_{og}^{(k)}(\bz,d)$ for which $\{\bX (w')|w'\in M_{og}^{(k)}(\bz,d),
w'\neq w\}\cap \cB (\bX (w)+\bs, (p+\etwo)n)$ is non-empty.
\end{lemma}

\begin{proof}
We fix a $k$ and prove both parts; a union bound over all $k$ at the end proves
the lemma.
{\em Part (i):}
Recall that the codewords in $\bX (M_{og}^{(k)}(\bz,d))$ are drawn
independently and uniformly from $\cSh (\bz,d)$. The proof of this part 
then follows using a two-step argument similar to the proof of 
Lemma~\ref{lem:list1}.

Part (i) implies that with high probability,
there are at most $c^2n^4$ codewords in each row of $M_{og}^{(k)}(\bz,d)$ which will
be confused, under the error vector $\bs$, with some codeword in another
row. The same statement also holds for columns by the same arguments.

{\em Part (ii):}
By part (i), with high probability over the code, there
are at most $c^2n^4$ codewords in any row (and column) which are confusable,
under the error vector $\bs$, with a codeword in another row.
Let us now define a directed graph with vertices $A^2$, and there is an edge
from $\bX (i,j)$ to $\bX (i',j')$ if $d_H(\bX (i,j)+\bs, \bX (i',j'))
\leq n(p+\etwo)$, that is, if the codeword $\bX (i,j)$ is confusable under
the error vector $\bs$
with $\bX (i',j')$. We define the non-horizontal out-degree of a node
as the number of edges coming out of that node to another node in a different
row. The non-vertical out-degree is similarly defined. Clearly the out-degree
of a node is at most the sum of the non-horizontal and non-vertical out-degree.
Under the high probability event of part (i),
each row of vertices in the graph has at most $c^2n^4$ vertices with non-zero
non-horizontal out-degree. So in the graph, there are at most $c^2n^4|A|$ nodes
with non-zero non-horizontal out-degree. Similarly there are at most $c^2n^4|A|$ nodes
with non-zero non-horizontal out-degree. So there are at most $2c^2n^4|A|
\leq 2^{n\hfour}$ 
nodes with non-zero out-degree.
\end{proof}

Lemma~\ref{lem:list1} and \ref{lem:list2} prove \eqref{eq:e3} for some $\beta>0$, 
for all $\bz,d,s,M$ satisfying $|d-qn|\leq \etwo n$ and
$2^{n(H(d/n)+R-1+\eone)}\geq |M|\geq 2^{n(H(d/n)+R-1-\eone)}$.
This in turn proves \eqref{eq:toshow}. Thus with high probability, 
the code satisfies $\evcf,\evcs$. Such a code then achieves exponentially
small probability of error by \eqref{eq:Perror}. This completes the achievability proof of Theorem~\ref{th:binarycapacity}.

%% file: proofs_gen4.tex
\def\fone{{f_1(p,q,\eone,\etwo)}}
 \def\ftwo{{f_2(q,p,\epsilon,\eone,\etwo)}}
 \def\fthree{{f_3(q,\epsilon)}}
 \def\ffour{{f_4(q,\epsilon)}}
 \def\ffive{{f_5(\etwo)}}
\def\fsix{{f_1(\etwo)}}
\def\fseven{{f_2(\etwo)}}
 \def\mosh{{\mathcal M}_s}
 \def\ogm{{\mathcal M}_{og}}
 \def\omosh{{\mathcal M}_o}

\newcommand{\googs}[2]{{\cC_0(#1,#2)}}
\newcommand{\gaau}[2]{{\cC_1(#1,#2)}}
\newcommand{\guuni}[2]{{\cC_2(#1,#2)}}

\newcommand{\bgOGS}{{\cC_0(\bZ, T)}}
\newcommand{\bgogs}{{\cC_0(\bz,\tau_{xz})}}
\newcommand{\bgAU}{{\cC_1(\bZ,T)}}
\newcommand{\bgau}{{\cC_1(\bz,\tau_{xz})}}
\newcommand{\bgUNI}{{\cC_2(\bZ,T)}}
\newcommand{\bguni}{{\cC_2(\bz,\tau_{xz})}}

\newcommand{\Sgcal}{{\cS_{cal}}}

 \def\cP{{\mathcal P}}

Let $T(\bx,\bz)$ denote the joint type of the vectors $\bx,\bz$.
Our achieability scheme is using a random code.
Let $p_X$ be the input distribution used to construct the code. 
We assume that it satisfies the condition in Theorem~\ref{th:g_lower}.
Recall that $p_{Z|X}$ is Calvin's channel law. These give a joint 
distribution $p_{XZ}$ and a marginal distribution $p_{Z}$. 
The channel law to Bob is given by $p_{Y|XS}$.
Note that every i.i.d. jamming strategy $p_{S|Z}$ results in a joint 
distribution $p_Xp_{Z|X}p_{S|Z}p_{Y|XS}$. 
We define
\begin{align*}
&\cW_{S|Z} := \left\{p_{S|Z} | \left(p_Xp_{Z|X}p_{S|Z}\right)_S \in \cW\right\},
\end{align*}
where $\left(p_Xp_{Z|X}p_{S|Z}\right)_S$ denotes the marginal of
$p_Xp_{Z|X}p_{S|Z}$ on $S$.
In the following, we assume the
distributions $p_X,p_{Z|X},p_{Y|XS}$ to be fixed.

For any $\by$, let us define $\cB_{X|Y}(\by, \cW)$
as the set of $\bx$ which are jointly $\etwo$-typical with
$\by$ for some $p_{S|Z}\in \cW_{S|Z}$. 
The volume of $\cB_{X|Y}(\by, \cW)$ can be bounded as
\begin{align*}
\frac{1}{n}\log_2 |\cB_{X|Y}(\by, \cW)| 
& \leq \max_{p_{S|Z} \in \cW_{S|Z}} H(X|Y) + \fsix 
\end{align*}
for some $\fsix \rightarrow 0$ as $\etwo \rightarrow 0$.
Similarly, $\cB_{Y|X}(\bx, \cW)$ is defined, and it has a volume
$\leq \max_{p_{S|Z} \in \cW_{S|Z}} H(Y|X) + \fsix$.
Drawing similarity with the proofs for $C(q,p)$, these sets take the role of 
balls of radius $p$ around $\bx$ and $\by$ respectively.

Decoding: On receiving $\by$, if Bob finds a unique codeword
in $\cB_{X|Y}(\by, \cW)$, then he decodes this codeword. Otherwise he
declares error. 

The overall proof argument is the same as in the bit-flip case.
So we only give the relevant modified lemmas and definitions in the
following, in addition to any extra arguments which are required
for the general case.

Let $\tau$ denote a joint type of $\bx,\bz$. In the following,
for a given $\bz$, $\tau$ denotes a joint type 
that is consistent with $\bz$.
\begin{align*}
\cSh(\bz, \tau) : = \{\bx\in \cX^n: T (\bx, \bz)=\tau\}.
\end{align*}
Let $T$ denote the type of $(\bX(W), \bZ)$.
We also denote by $\cT_{\mu} (X,Z)$, the set of joint types which are 
$\mu$-typical. Here $\mu=\mu (\etwo)$ (it is a
function of $p_X,p_{Z|X}$, though we do not mention it explicitly).

We define events, similar to those in the proof for the $C(q,p)$, 
\begin{align*}
& \evcf (\bz, \tau) :=  \left\{ 2^{n(H_\tau(X|Z)+R-H(X)+\eone)} \geq |\{w: X^n(w)\in \cSh(\bz,\tau)\}| \right. 
 \left. \geq 2^{n(H_\tau(X|Z)+R-H(X)-\eone)} \right\}\\
& \evcf := \cap_{\bz, \tau \in \cT_{\mu} (X,Z)}\evcf (\bz, \tau)\\
& \evd :=  \{T \in \cT_{\mu} (X,Z)\}.\\
\end{align*}

Here $H_\tau(X|Z)$ denotes the conditional entropy for the joint distribution
$\tau$. For every $\bz$, and $\tau \in \cT_{\mu} (X,Z)$, the oracle
partitions the set of messages
$$ M_s(\bz, \tau) = \{w: X^n(w)\in \cSh(\bz,\tau)\} $$
with codewords on the shell $\cSh(\bz,\tau)$ into disjoint
subsets $M_{og}^{(1)}(\bz, \tau), M_{og}^{(2)}(\bz, \tau), \cdots, M_{og}^{(\lambda(\bz, \tau))}(\bz, \tau)$
of size $2^{n\delta}$, except possibly the last subset with smaller size.
The oracle reveals $T$ and the particular subset
$\ogm = M_{og}^{(i)}(\bZ, T)$ of $M_s(\bZ, T)$ that contains the
encoded message $W$. 

In the general case under consideration, $\tau, T$ respectively
take the role of $d$ and $D$ in the binary case.
Like in the binary case, Lemmas~\ref{lemma:S0} and \ref{lemma:shell} also hold here with
the changed definitions of $\evd$ and $\evcf$. $h_2$ now depends on
$p_{Z|X}, \cW, \epsilon, \etwo, \eone$.
With $D$ replaced by $T$, and the changed definition of $\Scal
:=\cSh(\bZ, T)$, the observations O.1, O.2, O.3, and Lemma~\ref{lem:E2} 
in the previous section still hold. 

Suppose $\tau \in \cT_{\mu} (X,Z)$.
The event $\evcs(\bz,\tau,\bs)$ is defined, similarly as before, 
over the code construction that
every subset of size $2^{n\delta}$ in the ordered
partition of $M_s(\bz, \tau)$ has at most only $2^{n\hfour}$
messages which will undergo decoding error for the respective realizations
of $D, \bZ, \bS, W$. The event $\evcs$ is defined, also similarly as before,
as 
\begin{align*}
\evcs & = \cap_{\bz,\tau,\bs} \evcs(\bz,\tau,\bs)
\end{align*}
where the intersection is over $\tau \in \cT_{\mu} (X,Z)$, all feasible
$\bs$, and all $\bz$.

We need to show the counterpart of \eqref{eq:e3}:
\begin{align}
Pr(\evcs^c(\bz,\tau,\bs)|M_s(\bz,\tau)=M) \leq e^{-\beta n^2}. \label{eq:e3g}
\end{align}
for every
$M$ satisfying $2^{n(H_\tau(X|Z)+R-H(X)+\eone)}\geq |M|\geq 2^{n(H_\tau(X|Z)+R-H(X)-\eone)}$.
The overall proof argument is the same as in the binary case.
We first note the fact that for large enough $n$, for any $\bX\in \cSh (\bz, \tau)$ and $\bs$
satisfying $\cW$, $\bX \in \cB_{X|Y} (\bs(\bX),\cW)$. 
So the transmitted codeword always satisfies the decoding condition.
This is because,
(i) $X-Z-S$ forms a Markov chain, (ii) $(\bX,\bz)$ is $\etwo /2$-typical,
(iii) $(\bz, \bs)$ is $\etwo /2$-typical for some $p_{S|Z}\in \cW_{S|Z}$.
To see why the third statement is true, we note that 
$\bz$ is $\mu$-typical, and $\bs$ satisfies $\cW$. The conditional
distribution $p_{S|Z}:= T(\bz)T(\bs|\bz)/p_Z $ (at all points with
$p_Z(z)\neq 0$) is in $\cW_{S|Z}$, and $(\bz, \bs)$ is $\etwo
/2$-typical for this $p_{S|Z}$ if $\mu$ is 
small enough. 

We now give the counterparts of Lemma~\ref{lem:list1} and
Lemma~\ref{lem:list2} below. Together, they imply \eqref{eq:e3g}. But first
we give a generalization of Lemma~\ref{lem:basic1}: 

\begin{lemma}
\label{lem:basicgen}
Let $V \subset \cX^n$ be a subset of $\etwo$-typical sequences w.r.t. the
distribution $p_X$ with cardinality $|V| \leq 2^{nc}$. Then (i) the
probability $Pr_{p_X}(V) \leq 2^{-n(H(X)-c-\fseven)}$ for some $\fseven
\rightarrow 0$ as $\etwo \rightarrow 0$, (ii) if $R<H(X)-c-\fseven$, and
$2^{nR}$ vectors $\bX(w);
w=1,2,\cdots,2^{nR}$ are chosen independently with i.i.d. $\sim p_X$ components,
then for a sufficiently large constant $\alpha$, $Pr (|\{w:\bX (w) \in V\}|> cn^2) \leq e^{-\alpha n^2/6}$.
\end{lemma}
\begin{proof}
Part (i) follows as the probability of each $\etwo$-typical sequence is
$\leq 2^{-n(H(X)-\fseven)}$. Part (ii) follows using part (i) in a similar
way as Lemma~\ref{lem:basic1} using the Chernoff bound.
\end{proof}

\begin{lemma}
\label{lem:genlist1}
There exists $\beta >0$ such that for every $\bz$, $\tau \in \cT_{\mu} (X,Z)$, state vector $\bs $
introduced by Calvin, for every subset $M$ of messages with
$2^{n(H_\tau(X|Z)+R-H(X)+\eone)}\geq |M|\geq 2^{n(H_\tau (X|Z)+R-H(X)-\eone)}$,
conditioned on $M_s(\bz,\tau)=M$, 
with probability at least $1-e^{-\beta n^2}$
over the code, for every $i< \lambda(\bz, \tau)$, there are at most $c^2 n^4$
codewords from $\bX (M_{og}^{(i)}(\bz,\tau))$ for which there is a 
codeword from $\bX ((M_{og}^{(i)}(\bz,\tau))^c)$
which lies in $\cB_{X|Y}(\bs(\bX (w)), \cW)$.
\end{lemma}
\begin{proof}
Clearly, $(M_{og}^{(i)}(\bz,\tau))^c = M_s^c(\bz,\tau)\cup (M_s(\bz,\tau)\setminus M_{og}^{(i)}(\bz,\tau))$.
Note that the codewords in $(M_s(\bz,\tau)\setminus M_{og}^{(i)}(\bz,\tau))$ are
uniformly distributed in $\cSh (\bz, \tau)$.
The codewords of $M_s^c(\bz,\tau)$ are chosen according to 
$p_X^n$ conditioned on the subset $(\cSh (\bz, \tau))^c$.
We first prove the statement for codewords in $M_s^c(\bz,\tau)$.
We have two key steps:

1. For every realization of $\bX (M_{og}^{(i)}(\bz,\tau))$, we consider
$V=\cup_{w\in M_{og}^{(i)}(\bz,\tau)} \cB_{X|Y} (\bs(\bX(w)), \cW)$.
This satisfies $(1/n)\log_2|V|\leq \delta + \fsix +
\max_{p_{S|Z} \in \cW_{S|Z}} H(X|Y)$.
Since $\cSh (\bz, \tau)$ contains only a subset of one type of $\bx$,
for large enough $n$, under the probability measure $p_X^n$,
$Pr (\cSh (\bz, \tau)) \leq 1/2$. Thus by Lemma~\ref{lem:basicgen}(i), for
the codeword of any message $w\in M_s^c(\bz,\tau)$, the probability of it 
being chosen from $V$ is 
\begin{align*}
&Pr(V|(\cSh (\bz, \tau))^c) \\
& \leq 2\cdot 2^{-n(H(X)-\delta - \fsix -
\fseven -\max_{p_{S|Z} \in \cW_{S|Z}} H(X|Y))}\\
& = 2\cdot 2^{-n(\min_{p_{S|Z} \in \cW_{S|Z}}I(X;Y)-\delta - \fsix -
\fseven )}\\
\end{align*}
Thus, by the same proof as that 
of Lemma~\ref{lem:basicgen} (ii), with probability at least $1-e^{-\beta n^2}$,
there are at most $cn^2$ codewords from $\bX (M_s^c(\bz,\tau))$ 
in $V$ for sufficiently small $\delta, \etwo$. Here we have used the fact that
$R<\min_{p_{S|Z} \in \cW_{S|Z}}I(X;Y) -\delta - \fsix - \fseven $. 
So the same statement is also true over the whole random code, that is,
when $\bX (M_{og}^{(i)}(\bz,\tau))$ is also chosen randomly.

2. Now, over the random choice of $\bX (M_{og}^{(i)}(\bz,\tau)) $ from the
vectors in $\cSh (\bz, \tau)$,  with probability at least
$1-e^{-\beta n^2}$, for every $\bx$, there are at most $cn^2$ codewords
from $\bX (M_{og}^{(i)}(\bz,\tau))$ in 
$(\bs)^{-1}(\cB_{Y|X} (\bx, \cW))$. Here
$(\bs)^{-1}(\cB_{Y|X} (\bx, \cW)) := \{\bx': \bs (\bx') \in \cB_{Y|X} (\bx, \cW)\} $. This follows using Lemma~\ref{lem:basicgen}
because, for every feasible $\bs$, $(1/n)\log_2 |(\bs)^{-1}(\cB_{Y|X} (\bx, \cW))|
\leq \max_{p_{S|Z} \in \cW_{S|Z}}H(X|Y,S)
+ \max_{p_{S|Z} \in \cW_{S|Z}} H(Y|X) < H(X|Z)$.

Now the proof follows by taking the union bound over all $i<\lambda(\bz, \tau)$
as in the binary case (Lemma~\ref{lem:list1}).
A similar proof also works for $M_s(\bz,\tau)\setminus M_{og}^{(i)}(\bz,\tau)$ by noting that the volume of the shell is $\geq 2^{n(H_\tau (X|Z) -\eone /2)}$
(for large enough $n$), and the number of codewords on the shell is 
$|M| \leq 2^{n(H_\tau (X|Z) +R-H(X) +\eone )}$. 
\end{proof}

Similar to the binary case, let us consider the codewords in $M_{og}^{(k)}(\bz,\tau)$
arranged in a square and indexed by
a set $A\times A$ with $|A|= 2^{n\htwo/2}$.

\begin{lemma}
\label{lem:genlist2}
There exists $\beta >0$ such that, for every $\bz$,
$\tau \in \cT_{\mu} (X,Z)$, state $\bs $ introduced by Calvin,
for every subset $M$ of messages with
$2^{n(H_\tau(X|Z)+R-H(X)+\eone)}\geq |M|\geq 2^{n(H_\tau (X|Z)+R-H(X)-\eone)}$,
conditioned on $M_s(\bz,\tau)=M$,
with probability at least $1-e^{-\beta n^2}$
over the code, for every $k< \lambda(\bz, \tau)$, (i) for every $i$,
there are at most $c^2 n^4$ codewords in the $i$-th row
$\{\bX (i,j)|j\in A\}$ for which $\{\bX (i',j')|i'\neq i,j'\in A\}\cap 
\cB_{X|Y} (\bs (\bX (i,j)), \cW)$ is non-empty,
(ii) there are at most $2^{n\hfour}$ messages $w\in
M_{og}^{(k)}(\bz,\tau)$ for which $\{\bX (w')|w'\in M_{og}^{(k)}(\bz,\tau),
w'\neq w\}\cap \cB_{X|Y} (\bs(\bX (w)), \cW)$ is non-empty.
\end{lemma}
\begin{proof}
The proof of the first part follows using a two-step argument similar
to the proof of Lemma~\ref{lem:genlist1}. The proof of the second part
 follows using similar arguments as the second part of Lemma~\ref{lem:list2}.
\end{proof}
Lemma~\ref{lem:genlist1} and \ref{lem:genlist2} prove \eqref{eq:e3g},
and then by the same arguments as in the bit-flip case, the achievability
of Theorem~\ref{th:g_lower} follows.

\begin{remark}[Improvements for bit erasing and flipping adversaries]
\label{rem:erasure-erasure}
Condition (c) on an achievable rate in Theorem~\ref{th:g_lower} is due to 
the bound
\begin{eqnarray}
(1/n)\log_2 |(\bs)^{-1}(\cB_{Y|X} (\bx, \cW))| 
\leq \max_{p_{S|Z} \in \cW_{S|Z}}H(X|Y,S)
+ \max_{p_{S|Z} \in \cW_{S|Z}} H(Y|X)
\label{eq:bound}
\end{eqnarray}
used in the second part of the proof of Lemma~\ref{lem:genlist1}. 
In general, \eqref{eq:bound} is the best bound we have in single-letter
expression for
$(1/n)\log_2 |(\bs)^{-1}(\cB_{Y|X} (\bx, \cW))|$. We suspect this
to be quite loose, and expect the result (Theorem~\ref{th:g_lower}) for
general channels to hold under weaker conditions than (c).

(i) Binary erasure-erasure channel: For the erasure-erasure channel
$CE(q,p)$, this condition gives $p+H(p) < q$, which is stronger than the
natural `sufficiently myopic' condition $p<q$. For this channel, it is
possible to get a tighter bound which results in the condition $p<q$,
that is the same as condition (b). We first note that 
$|(\bs)^{-1}(\cB_{Y|X} (\bx, \cW))|$
counts the number of vectors $\bx'$ for which $\bs (\bx')$ can also
be obtained by erasing some components of $\bx$. The components of
$\bx$ erased in this process must be the same as those indicated by $\bs$,
that is, $\bs'(\bx) = \bs(\bx')$ only if $\bs = \bs'$.
Thus, 
\begin{align*}
& (1/n)\log_2 |(\bs)^{-1}(\cB_{Y|X} (\bx, \cW))| \leq p,
\end{align*}
as the number of erasures in $\bs$ is at most $np$.
Thus under condition (b) alone, i.e., $p<q$, the capacity
of $CE(q,p)$ is $1-p$. This proves Theorem~\ref{th:erasurecapacity}.

(ii) Bit erasing and flipping adversary, $CEF(p_{Z|X}, p_e,p_w)$: We note that the counting of $|(\bs)^{-1}(\cB_{Y|X} (\bx, \cW))|$
only involves the channel between Alice to Bob (the AVC), and not the channel between Alice to James.
So similar improved counting as for the erasure-erasure channel
works as long as James can only erase and
flip transmitted bits.
For such adversaries, irrespective of James's own
channel, Theorem~\ref{th:g_lower} holds without the extra
condition (c).

In particular, let us consider the setup $CEF(p_{Z|X}, p_e,p_w)$ where James
can erase upto $p_e$ fraction of the transmitted bits, and he can flip upto
$p_w$ fraction of the transmitted bits. For such a valid state vector
$\bs$, let  $s^n_e$ denote the action of erasing the same positions that
are erased by $\bs$. Now clearly,
\begin{align*}
& (\bs)^{-1}(\cB_{Y|X} (\bx, \cW)) = \{\bx'|d_H(s^n(\bx'),s^n_e(\bx)) 
\leq p_wn\}.
\end{align*}
Thus 
\begin{align*}
& (1/n)\log_2 |(\bs)^{-1}(\cB_{Y|X} (\bx, \cW))| \leq 
p_e + (1-p_e) H\left(\frac{p_w}{1-p_e}\right).
\end{align*}
This gives an achievable rate upto
\begin{align*}
& 1 - \left(p_e + (1-p_e) H\left(\frac{p_w}{1-p_e}\right)\right)\\
= & (1-p_e)\left(1-H\left(\frac{p_w}{1-p_e}\right)\right)
\end{align*}
under the sufficient myopicity condition that 
\begin{align*}
1-H(X|Z) \leq (1-p_e)\left(1-H\left(\frac{p_w}{1-p_e}\right)\right)
\end{align*}

(iii) Secrecy capacity for the  erasing and flipping adversary, $CEF(p_{Z|X}, p_e,p_w)$: 
Using a standard stochastic encoding technique
for wiretap channels, the above also gives a secrecy rate of 
\begin{align*}
& (1-p_e)\left(1-H\left(\frac{p_w}{1-p_e}\right)\right) - (1- H(X|Z))\\
= & H(X|Z)+p_eH\left(\frac{p_w}{1-p_e}\right) - p_e-H\left(\frac{p_w}{1-p_e}\right).
\end{align*}
For the special case of $CEF(BEC(q), p_e,p_w)$, this
gives the secrecy rate
\begin{align*}
& (1-p_e)\left(1-H\left(\frac{p_w}{1-p_e}\right)\right) - (1-q)\\
= & q+p_eH\left(\frac{p_w}{1-p_e}\right) - p_e-H\left(\frac{p_w}{1-p_e}\right).
\end{align*}

(iv) Secrecy against a type II wiretapper adversary: We now elaborate on the erasure-erasure setting and explicitly on extending Theorem~\ref{th:erasuresecrecy} to the model of Aggarwal et al.~\cite{aggarwal2009wiretap}.
In general, the proof of Theorem~\ref{th:erasuresecrecy} follows that of Theorem~\ref{th:g_secrecy} when one considers the channel $CE(q,p)$. 
Specifically, the proof of secrecy for rates under $q - p$ follows roughly by appending $1-q$ bits of private randomness $r$ to the message $u$, and applying our code on the concatenated pair $(u,r)$ (which is of length $(1-p)n$). To prove that  our random code construction satisfies the secrecy requirement $\frac{1}{n}I(\bZ;U)<\e$ in the model of \cite{aggarwal2009wiretap} (when James can pick which $qn$ bits are erased) one must show that with high probability over code design, for every view $\bZ$ of James, for every $u \in U$ there are approximately the same number of random strings $r$ such that the encoding of $(u,r)$ is consistent with $\bZ$ (i.e., the corresponding codeword agrees with $\bZ$ on all the un-erased entries). Indeed, in this case, the amount of information James has on $U$ is limited. Let $\e>0$. Taking the size of $r$ to be $1-q+\e$ and the rate of $u$ to be $q-p-2\e$, we have (via the Chernoff bound) with doubly exponential probability of $1-2^{-2^{\Omega(\e n)}}$ over code design that for every view $\bZ$ of James and every $u \in U$ 
the number of different random $r$ such that the encoding of $(u,r)$ is consistent with $\bZ$ is in the range $[2^{\e n}(1-\e), 2^{\e n}(1+\e)]$.
Bounding entropy by variational distance we conclude that $\frac{1}{n}I(\bZ;U)<O(\e)$.

(v) Wiretap channel of type II with erasing and flipping adversary,
$\wcefII(p_r,p_e,p_w)$: If James can choose a $p_r = 1-q$ fraction of transmitted bits
to observe,
then using the lines of argument outlined above, it can be 
seen that the secrecy rate $q+p_eH\left(\frac{p_w}{1-p_e}\right) - p_e-H\left(\frac{p_w}{1-p_e}\right) = 1-p_r+p_eH\left(\frac{p_w}{1-p_e}\right) - p_e-H\left(\frac{p_w}{1-p_e}\right)$ is still achievable.
\end{remark}